\documentclass[aps,prb,twocolumn,superscriptaddress]{revtex4-2}
\usepackage{hyperref}
\hypersetup{
    colorlinks=true,
    linkcolor=blue,
    filecolor=blue,
    urlcolor=blue,
    citecolor=blue
}
\usepackage{amssymb,amsthm,mathtools,mathrsfs,cleveref,algorithmicx,algpseudocode,framed,color}
\usepackage[normalem]{ulem}

\newtheorem{thm}{Theorem}\crefname{thm}{Theorem}{Theorems}
\newtheorem*{thm*}{Theorem}
\newtheorem*{cor*}{Corollary}
\newtheorem{lem}[thm]{Lemma}\crefname{lem}{Lemma}{Lemmas}
\crefname{prp}{Proposition}{Propositions}
\newtheorem{cor}[thm]{Corollary}\crefname{cor}{Corollary}{Corollaries}
\crefname{dfn}{Definition}{Definitions}
\crefname{section}{Section}{Sections}
\crefname{appendix}{Appendix}{Appendices}

\newcommand{\R}{\mathbb{R}}

\newcommand{\C}{\mathbb{C}}

\begin{document}


\title{Free Mode Removal and Mode Decoupling for \texorpdfstring{\\}{ } Simulating General Superconducting Quantum Circuits}
\author{Dawei Ding}
\email{d.ding@alibaba-inc.com}
\affiliation{Alibaba Quantum Laboratory, Alibaba Group USA, Bellevue, Washington 98004, USA}

\author{Hsiang-Sheng Ku}
\email{hsiang-sheng.gxs@alibaba-inc.com}
\affiliation{Alibaba Quantum Laboratory, Alibaba Group, Hangzhou, Zhejiang 311121, P.R.China}

\author{Yaoyun Shi}
\affiliation{Alibaba Quantum Laboratory, Alibaba Group USA, Bellevue, Washington 98004, USA}

\author{Hui-Hai Zhao}
\email{huihai.zhh@alibaba-inc.com}
\affiliation{Alibaba Quantum Laboratory, Alibaba Group, Beijing 100102, P.R.\ China}
\date{\today}

\begin{abstract}
  Superconducting quantum circuits is one of the leading candidates for a universal quantum computer. Designing novel qubit and multiqubit superconducting circuits requires the ability to simulate and analyze the properties of a general circuit. In particular, going outside the transmon approach, we cannot make assumptions on anharmonicity, thus precluding blackbox quantization approaches and necessitating the formal circuit quantization approach. We consider and solve two issues involved in simulating general superconducting circuits. One of the issues is the handling of free modes in the circuit, that is, circuit modes with no potential term in the Hamiltonian. Another issue is circuit size, namely the challenge of simulating strongly coupled multimode circuits. The main mathematical tool we use to address these issues is the linear canonical transformation in the setting of quantum mechanics. We address the first issue by giving a provably correct algorithm for removing free modes by performing a linear canonical transformation to completely decouple the free modes from other circuit modes. We address the second by giving a series of different linear canonical transformations to reduce intermode couplings, thereby reducing the problem to the weakly coupled case and greatly mitigating the overhead for classical simulation. We benchmark our decoupling methods by applying them to the circuit of two inductively coupled fluxonium qubits, obtaining several orders of magnitude reduction in the size of the Hilbert space that needs to be simulated.
\end{abstract}

\maketitle

\section{Introduction}
Superconducting quantum circuits is a promising physical medium for realizing universal quantum computing. Progress in coherence times, fabrication, and qubit designs have propelled the superconducting path to achieving the first useful and scalable quantum computer. For reviews about current state of the art results and methods, see for instance~\cite{kjaergaard2020superconducting} and~\cite{krantz2019quantum}.

The first step in designing superconducting quantum circuits is to classically simulate circuits of individual qubits or a number of coupled qubits. This is done by finding a Hamiltonian that models the properties of the circuit and its dynamics. However, a general superconducting circuit can be complex, containing multiple qubit modes with various anharmonicities depending on the type of qubit we are considering. These qubit modes can be coupled to each other and to linear harmonic modes for readout and control. Thus, classically simulating general circuits can be a nontrivial problem and numerically costly without proper techniques.

A method to address this for circuits that are linear or with low anharmonicity is black box quantization~\cite{manucharyan2007microwave,nigg2012black}. This technique uses the harmonic eigenmodes obtained from the electromagnetic simulation of the linear response function of the circuit as the basis for the Hamiltonian. This is sufficient for circuits consisting of weakly anharmonic transmons~\cite{koch2007charge}, but not for more general superconducting circuits.

Hence, to enlarge our search for novel circuit designs, we need to analyze the Hamiltonian of a general superconducting quantum circuit. Such Hamiltonians are given in terms of canonically conjugate variables of fluxes and charges, respectively position and momentum, of various circuit components, or linear combinations thereof~\cite{yurke1984quantum,devoret1995quantum,burkard2004multilevel,burkard2005circuit}. The parameters in the Hamiltonians are functions of linear inductances and capacitances of the circuit, which can be obtained by a series of electromagnetic simulations, in addition to characteristic energies of Josephson junctions. However, there are two general issues encountered in analyzing such Hamiltonians. 

One of the issues is the handling of free modes in the circuit, that is, circuit modes with no potential term in the Hamiltonian. These modes appear for instance in circuits which are floating with respect to ground. Other examples are when there are multiple components of the circuit which are only capacitively coupled. For a general circuit, such free modes can have charge coupling with other modes. However, we mathematically prove that for a general circuit, these modes can be transformed to be fully independent of the other modes. The transformed free modes do not participate in the dynamics of the circuit, which usually involve qubit and resonator modes. We then show how to correctly extract the Hamiltonian for these nonfree modes that actually participate in quantum information processing and measurement.

Another issue is simply one of size. The simulation of strongly coupled multimode circuits is quite challenging --- indeed, this complexity is the main reason quantum computers are interesting in the first place. Na\"ive approaches to numerically analyzing strongly coupled multimode circuits can easily become computationally prohibitive, and it is necessary to find improved methods with reduced computational costs that can be applied to general circuits. Some progress is made in~\cite{rodriguez_2014}, where the author provides a software package to quantize general superconducting circuits and analyze the Hamiltonian obtained. There, the author makes a transformation to perform the Born-Oppenheimer approximation to separate the fast and slow degrees of freedom. There are also studies on multimode circuits consisting of arrays of Josephson junctions~\cite{junction_chain1997,junction_ladder2003,current_mirror_DMRG2019,fluxonium_dmrg2019} that have employed the density matrix renormalization group (DMRG) method~\cite{dmrg1992,dmrg1993} to efficiently and accurately compute the low-lying eigenstates. These simulations are performed on many-mode Hamiltonians constructed by transformations of specific types of circuits. Hence, both of these methods are not sufficiently general.

We address both of these issues using the mathematical tool of the linear canonical transformation in the setting of quantum mechanics. This tool has been extensively studied in the literature~\cite{itzykson1967remarks,bargmann1970group,moshinsky1971linear} and is extremely useful for our purposes since the Hamiltonian of general superconducting circuits is mostly quadratic in the charge and flux operators. It is therefore natural to consider linear transformations to manipulate the form of the Hamiltonian. In particular, we will use linear canonical transformations to fully decouple the free modes from the other modes, thereby allowing us to directly remove them from the Hamiltonian. We will also use linear canonical transformations to reduce couplings between different circuit modes, allowing us to reduce the problem of simulating strongly coupled multimode circuits to the weakly coupled case. This significantly reduces the complexity of numerically analyzing the overall circuit Hamiltonian. We again emphasize that our techniques do not make any assumptions about the superconducting circuit save minimal requirements for circuit quantization~\footnote{Even these requirements can often be satisfied by adding fictitious small capacitors or large inductors. } and are otherwise completely general. This need to simulate general circuits comes at a time where, despite the success of the now standard transmon qubit, the search for novel qubit and multiqubit designs is still very much ongoing. Recent successful experimental demonstrations of the fluxonium~\cite{nguyen2019high} and $0$-$\pi$~\cite{gyenis2019experimental} qubits are some examples. 

We arrange the paper as follows. For completeness, we first review the quantization methods for general superconducting circuits as well as linear canonical transformations in quantum mechanics in~\cref{sec:preliminaries}. In~\cref{sec:free_mode}, we introduce the methods to identify and remove free modes, applying them to the circuit of a Cooper-pair box with a drive and capacitively coupled to a resonator. In~\cref{sec:decouple}, we introduce three mode decoupling techniques and compare their performance by applying them to the circuit of two inductively coupled fluxonium qubits. Finally, we summarize with a discussion of future work in~\cref{sec:conclusion}. In general, we structure our paper to obviate the need to dive into the proofs and mathematical details, and those interested only in implementation can readily obtain the relevant information as well as high-level explanations of the methods.

\section{Preliminaries}\label{sec:preliminaries}
We cite and derive some results that we will use in the rest of the paper. 

\subsection{Quantization of a General Superconducting Circuit}
The Hamiltonian of a general superconducting circuit can be obtained using techniques from network theory~\cite{yurke1984quantum,devoret1995quantum,burkard2004multilevel,burkard2005circuit}. In the most recent approaches~\cite{burkard2004multilevel,burkard2005circuit}, a superconducting circuit is treated as a multigraph, where the nodes are the vertices and the circuit elements are the edges. When there is no dissipation, the circuit elements can be capacitors, inductors, Josephson junctions, and voltage or current sources. They then find a spanning tree consisting of a subset of classes of elements (e.g.\ Josephson junctions and inductors) and show that the fluxes across the edges of the spanning tree form a set of generalized coordinates for the system. Finally they derive the equations of motion for the fluxes in the spanning tree and thereby obtain the Lagrangian. A Legendre transformation is performed to obtain the classical Hamiltonian, in the process identifying linear combinations of the charges across the edges as the variables conjugate to the fluxes. The quantum Hamiltonian is obtained via canonical quantization.

In our paper, we use the approach in~\cite{burkard2005circuit} since we often have circuits with a large number of capacitors, and the other approach~\cite{burkard2004multilevel} does not allow for circuits with capacitor loops. With this approach, assuming the spanning tree has $n$ edges containing an inductor or Josephson junction, the Hamiltonian takes the form
\begin{align}
  \label{eq:circ_ham}
  H = & \frac{1}{2} (\vec Q - C_V \vec V)^T  \mathcal{C}^{-1} (\vec Q - C_V \vec V) \nonumber \\
  & - \sum_{i=1}^{n_J} E_{J,i} \cos \left(\frac{2 \pi}{\Phi_0} \Phi_i \right) + \frac{1}{2} \vec \Phi^T M_0 \vec \Phi + \vec \Phi^T N \vec \Phi_x,
\end{align}
where $\vec \Phi = (\Phi_1, \Phi_2, \cdots, \Phi_n)^T$ is the vector of flux operators for the different edges of the spanning tree, $\vec Q = (Q_1, Q_2, \cdots, Q_n)^T$ is the vector~\footnote{An important clarification is needed here. The correct interpretation of~\cref{eq:circ_ham} and similar expressions is to perform all matrix multiplications assuming all the quantum operators are scalars. Polynomials of $\Phi_i, Q_i$ are obtained, at which point the quantum operator nature can be reintroduced. Each term in the polynomials only involves commuting operators, so there is no ambiguity in the operator ordering. } of canonically conjugate operators corresponding to linear combinations of the charges across the edges, $\vec V$ is the vector of voltage biases in the circuit each multiplied with the identity operator, $\mathcal{C}^{-1}, M_0, N, C_V$ are the matrices of charge, flux, external flux, and voltage couplings, respectively, $\vec \Phi_x$ is the external magnetic flux, $n_J$ is the number of Josephson junctions with $E_{J,i}$ being the characteristic energy scale of each junction, and $\Phi_0 \equiv \frac{h}{2e}$ is the superconducting flux quantum. The flux operators correspond to edges in the spanning tree containing either a Josephson junction or an inductor. In our work we will often distinguish between (Josephson) junction modes and inductor modes.

\subsection{Linear Canonical Transformations in Quantum Mechanics}
Looking at the Hamiltonian in~\cref{eq:circ_ham}, the flux and charge operators mainly appear as quadratic forms, that is, expressions consisting entirely of terms quadratic in $Q$ or $\Phi$. In particular, all the coupling terms are quadratic forms. With this observation, it is natural to consider linear transformations~\cite{itzykson1967remarks,bargmann1970group,moshinsky1971linear} of the flux and charge operators to manipulate the form of the Hamiltonian, and this is the main mathematical tool we use in this work. We consider a special case where we do not mix flux and charge operators. We will see that this is sufficient for our purposes when \emph{a priori} a more general transformation may be needed. More specifically, for us a linear canonical transformation is described by a real invertible matrix $W \in \mathrm{GL}(n, \R)$, which transforms the canonical flux and charge operators as 
\begin{align}
  \label{eq:lin_can}
  \vec \Phi \mapsto \vec \Phi' \equiv W \vec \Phi, \, \vec Q \mapsto \vec Q' \equiv (W^T)^{-1} \vec Q.
\end{align}
This linear transformation preserves the canonical commutation relations:
\begin{align}
  [ \Phi_i, \Phi_j ]= 0, \, [ Q_i, Q_j ] =0, \, [ \Phi_i, Q_j ] = i \hbar \delta_{ij}.
\end{align}
That the first two types are preserved is easily checked. For the third, the following computation can be done:
\begin{align}
  \left[ \Phi_i', Q_j' \right] & = \left[ \sum_{k=1}^{n} W_{ik} \Phi_k, \sum_{l=1}^{n} W^{-1}_{lj} Q_l \right] \nonumber\\
  & = \sum_{k,l = 1}^{n} W_{ik} W^{-1}_{lj} i \hbar \delta_{kl} \nonumber\\
  & = i\hbar \delta_{ij}.
\end{align}

The main fact about linear canonical transformations that we will use is that the Hamiltonian in terms of the transformed operators can be numerically analyzed in exactly the same way as before the transformation. More precisely, the matrices assigned to $\vec \Phi, \vec Q$ to discretize the Hamiltonian can be used for $\vec \Phi', \vec Q'$. Although the resulting matrix forms of the Hamiltonians will not be the same, observable properties such as the spectrum and matrix elements between different energy eigenstates will be the same. This is because the two matrices are unitarily equivalent. That is, letting $H_1, H_2$ be the two matrices, there exists a unitary $U$ such that
\begin{align}
  H_2 = U H_1 U^\dagger.
\end{align}

In more abstract terms, a linear canonical transformation is a \textit{symplectic transformation}, which is defined as a real matrix $S \in \R^{2n \times 2n}$ that preserves the symplectic form:
\begin{align}
  S^T \Omega S = \Omega,
\end{align}
where the symplectic form is defined as the matrix
\begin{align}
  \Omega \equiv
  \begin{pmatrix}
    0_n & I_n \\
    -I_n & 0_n
  \end{pmatrix},
\end{align}
where $0_n$ is the $n\times n$ dimensional zero matrix and $I_n$ is the $n\times n$ identity matrix. In our case, the symplectic matrix is block diagonal:
\begin{align}
  S =
  \begin{pmatrix}
    W & 0_n \\
    0_n & (W^T)^{-1}
  \end{pmatrix} \in \mathrm{Sp}(2n, \R).
\end{align}
This matrix acts on the $2n$-dimensional vector of operators
\begin{align}
  \begin{pmatrix}
    \vec \Phi\\
    \vec Q
  \end{pmatrix}.
\end{align}
The fact about linear canonical transformations that we need can be found in~\cite{moshinsky1971linear}, which we express informally by the following lemma: 
\begin{lem}
  Let $(x_1, x_2, \cdots, x_n, p_1, p_2, \cdots, p_n)^T$ be a vector of canonical quantum operators. Furthermore, let $S \in \mathrm{Sp}(2n, \R)$ be a symplectic matrix and define the operators
  \begin{align}
    & (x_1', \cdots, x_n', p_1', \cdots, p_n')^T \nonumber\\
    & \equiv S (x_1, x_2, \cdots, x_n, p_1, p_2, \cdots, p_n)^T.
  \end{align}
  Then, there exists a unitary $U$ such that for all $i \in 1, \cdots, n$,
  \begin{align}
    x_i' = U x_i U^\dagger
  \end{align}
  and
  \begin{align}
    p_i' = U p_i U^\dagger.
  \end{align}
  \label{lem:lin_can}
\end{lem}
That is, we can effectively treat the transformed fluxes and charges in the same way as the originals. Hence, we can manipulate the form of the Hamiltonian via the corresponding transformations on the coupling matrices in~\cref{eq:circ_ham} when we replace $\vec \Phi, \vec Q$ with $\vec \Phi', \vec Q'$. The transformations of the flux and charge coupling matrices are given respectively by
\begin{align}
  \label{eq:coupling_transform}
  M_0 \mapsto (W^T)^{-1} M_0 W^{-1}, \, \mathcal{C}^{-1} \mapsto W \mathcal{C}^{-1} W^T.
\end{align}
The other coupling matrices transform as
\begin{align}
  N \mapsto (W^T)^{-1} N, \, C_V \mapsto (W^T)^{-1} C_V.
\end{align}
In general, the junction term will be transformed as well, but this will be applicable only for one of the techniques we provide to reduce coupling between modes described in~\cref{sec:full_symplectic}.

\section{Free Mode Removal} \label{sec:free_mode}
Free modes are degrees of freedom for which only the charge operator appears in the Hamiltonian. They can occur for instance if the circuit is floating. Taking the fluxes to be the generalized coordinates and the charges to be the conjugate momenta, these modes have no potential, hence the name. We show that free modes can be transformed to be completely independent of the other modes. We can also easily adapt this technique to the dual scenario where only the flux operator appears for certain modes. 

Free modes cannot be simply removed directly from the Hamiltonian since they can be charge coupled to the other modes. Another complication is that free modes are not always explicit in the Hamiltonian and a transformation of the modes may be necessary to identify and separate them. Consider~\cref{eq:circ_ham}. The number of free modes is given by
\begin{align}
  \label{eq:num_free}
  F \equiv \dim(\ker(M_0) \cap \ker(N^T) \cap V_L),
\end{align}
where $V_L$ is the subspace spanned by inductor fluxes~\footnote{Note that in~\cite{burkard2005circuit} the conditions on the circuit are such that the modes are never perfectly free but have a vanishingly small potential term. In this case a small threshold is set instead. See the example for more details. }. More explicitly, in the formulation of~\cite{burkard2005circuit} the fluxes $\Phi_i$ in $\vec \Phi$ correspond to fluxes across either junctions or inductors, and $V_L$ is defined as the subspace of vectors $\vec v$ for which the elements $v_i$ corresponding to junction fluxes (the same correspondence as $\vec \Phi$) are zero. Now, we need to perform an appropriate linear canonical transformation on our modes to make the free modes explicit. This would be a orthogonal transformation $W$ on $\vec \Phi$ so that flux operators form a basis for the subspace in~\cref{eq:num_free}. By~\cref{eq:lin_can}, the same transformation $W$ will be applied to $\vec Q$. Then, we can partition our modes as $\vec \Phi = (\Phi_1, \cdots, \Phi_F, \cdots \Phi_n)$, where $\Phi_1, \cdots, \Phi_F$ are the flux operators for the free modes, and similarly for $\vec Q$. Mathematically, this means after the transformation the first $F$ rows of $N$ and the first $F$ rows and columns of $M_0$ are all zeros. In this form the free modes can be identified and separated.

The proposed approach to remove free modes is to linearly transform the modes to implement Gaussian elimination on $\mathcal{C}$, the effective capacitance matrix of the circuit derived in~\cite{burkard2005circuit} and the inverse of the charge coupling matrix $\mathcal{C}^{-1}$. This is somewhat counterintuitive since the problem is the charge coupling between free modes and other modes, and so it may seem more natural to implement Gaussian elimination directly on $\mathcal{C}^{-1}$. However, we choose to do it on $\mathcal{C}$ since~\cref{eq:coupling_transform} implies that \textit{the corresponding transformation on $M_0$ will implement the same Gaussian elimination}. By definition, the free modes have vanishing potential, and so the flux observables do not appear in the Hamiltonian. Thus, the Gaussian elimination procedure will preserve $M_0$ and $N$, the flux coupling matrices. This ensures that the free modes are still free after the transformation and not mixed with other modes. Furthermore, removing off-diagonal terms in $\mathcal{C}$ for the free modes also removes them for $\mathcal{C}^{-1}$. Hence, after our transformation we obtain an explicit set of free modes which are decoupled from all other modes, allowing us to remove them at will.

The proofs for why our algorithm can be run and is correct are technical but standard Gaussian elimination, and we include the details in~\cref{sec:gauss_elim}. For purely implementation purposes, the details can be skipped save for how to generate the $W$ matrix, which is everything up to~\cref{eq:free_W}. Interestingly, we find that the final Hamiltonian of the nonfree modes is the same as that of just removing the terms involving free modes except for the voltage coupling matrix, whose final form depends on the linear canonical transformation matrix $W$ that we obtain in this algorithm. In particular, the flux and charge couplings between nonfree modes are preserved, and the flux operators of the nonfree modes are not transformed, thereby preserving the junction term in the Hamiltonian. Hence, procedurally it is sufficient to compute the $W$ matrix as outlined below and transform $C_V$ via
\begin{align}
  C_V \mapsto (W^T)^{-1} C_V.
\end{align}
Then, the correct Hamiltonian for the other modes is obtained by removing all the terms involving the free modes, including charge coupling terms.


\subsection{Example}
We give a detailed example of removing free modes in a circuit. Consider a Cooper-pair box (Josephson junction and capacitor in parallel) with a drive and capacitively coupled to a resonator (LC circuit). A circuit diagram is shown in~\cref{fig:qubit_resonator}.
\begin{figure}
  \includegraphics[width = 0.48\textwidth]{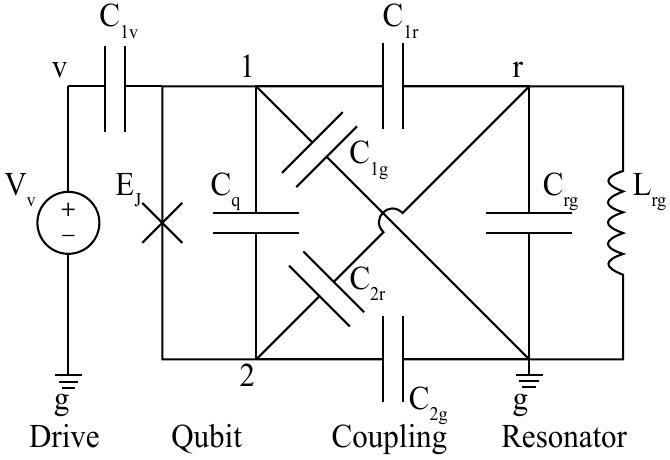}
  \caption{Circuit diagram of a Cooper-pair box with a drive and capacitively coupled to a resonator. The spanning tree edges are $[1,2], [g,1], [r,g], [g,v]$. \label{fig:qubit_resonator}}
\end{figure}
The circuit elements and the nodes are labeled for reference. 

We can give fictitious numerical values for all the circuit parameters. The capacitances are given in~\cref{tab:caps}.
\begin{table}
  \caption{Capacitance table. Every edge is defined by the pair of nodes with labels defined in~\cref{fig:qubit_resonator}. The capacitances on the edges not listed in the table are neglected.
  \label{tab:caps}}
  \begin{ruledtabular}
  \begin{tabular}{|c|c|c|c|c|c|c|c|}
    \hline
    Edge & $[1, g]$ & $[2,g]$ & $[r,g]$&$[1, 2]$&$[1,r]$&$[2,r]$ &$[1,v]$ \\
    \hline
    Capacitance (fF) & 1 & 2 & 30 & 40 & 5 & 6 & 0.007\\
    \hline
  \end{tabular}
  \end{ruledtabular}
\end{table}
We also take $L_{rg} = 800$nH, and $E_J = 3$GHz$\cdot h$.

The Hamiltonian of this circuit is
\begin{align}
  H = &\frac{1}{2} (\vec Q - C_V \vec V)^T  \mathcal{C}^{-1} (\vec Q- C_V \vec V) - E_J \cos \left(\frac{2 \pi}{\Phi_0} \Phi_q \right) \nonumber\\
  & + \frac{1}{2} L_{rg}^{-1} \Phi_{rg}^2,
\end{align}
where $E_J = 3$GHz$\cdot h$, $L_{rg}^{-1} = 0.0204$GHz$\cdot h$, and
\begin{align}
  \vec \Phi = (\Phi_{1g}, \Phi_{q}, \Phi_{rg})^T,
  \vec Q = (Q_{1g}, Q_{q}, Q_{rg})^T.
\end{align}
The subscripts describe which branch of the circuit the observable corresponds to, and we will also use it to denote modes and their corresponding subspaces. Note that for this particular circuit the charge operators are exactly the charges across the edges of the spanning tree, instead of linear combinations. The coupling matrices are given by
\begin{align}
  \mathcal{C}^{-1} = 
  \begin{pmatrix}
    15.2 & -2.06 & 3.78\\
    -2.06 & 3.57 & -0.0312\\
    3.78 & -0.0312 & 4.79
  \end{pmatrix}
\end{align}
and
\begin{align}
  M_0 = 
  \begin{pmatrix}
    0 & 0 & 0\\
    0 & 0 & 0\\
    0 & 0 & 0.0204
  \end{pmatrix}.
\end{align}
Throughout we keep three significant figures and use GHz$\cdot h$ as our units for energy. We use energy units for $M_0, \mathcal{C}^{-1}$ by replacing $\Phi,Q$ respectively with the dimensionless operators $\phi,n$ where
\begin{align}
  \Phi = \frac{\hbar}{2e} \phi, Q = 2e n.
\end{align}
The number of free modes is 1 and is here the $1g$ mode since the inductor subspace is spanned by the $1g$ and $rg$ subspaces~\footnote{Note that we actually add an inductor with a very large inductance on the $1g$ edge which we do not include in the diagram for simplicity. This is to satisfy the conditions for circuit quantization in~\cite{burkard2005circuit}, but the mode will be flagged as free since its eigenvalue for $M_0$ is below a small threshold that we set. Our algorithm then sets the corresponding coefficient of the flux operator to zero. This procedure treats the large inductor as a placeholder and does not take the infinite inductance limit as outlined in~\cite{koch2009charging}. }, and the kernel of $M_0$ is spanned by the $q$ and $1g$ subspaces. We also give the voltage coupling matrix 
\begin{align}
  C_V = 
  \begin{pmatrix}
    -21.8\\
    0\\
    0
  \end{pmatrix}.
\end{align}
Only the transformation of this matrix is nontrivial. The units for this matrix is Farad/(2e), which is a unit of inverse voltage.


Now, it is clear that we are already in a basis where the free mode appears: the $1g$ mode has no potential term. Hence, we can directly go to Gaussian elimination. Since there is only one free mode, by our definition given in~\cref{eq:W_def}, we can obtain the $W$ matrix:
\begin{align}
  (W^T)^{-1} = W_1 =
  \begin{pmatrix}
    -1 & 0 & 0\\
    -0.571 & 1 & 0\\
    0.785 & 0 & 1
  \end{pmatrix}.
\end{align}
A direct calculation yields
\begin{align}
  \mathcal{C}'^{-1} &= W \mathcal{C}^{-1} W^T \nonumber \\
  & = 
  \begin{pmatrix}
    11.1 & 0 & 0\\
    0 & 3.57 & -0.0312\\
    0 & -0.0312 & 4.79
  \end{pmatrix}.
\end{align}
As expected, $\mathcal{C}'^{-1}$ is block diagonal and the charge couplings of the nonfree modes are preserved. The transformed modes are given by
\begin{eqnarray}
  &   & (n_1', n_2', n_3') \nonumber\\
  & = & ( (W^T)^{-1} \vec n)^T \nonumber\\
  & = & (- n_{1g}, -0.571 n_{1g} + n_q, 0.785 n_{1g}+n_r),
\end{eqnarray}
\begin{eqnarray}
  & & (\phi_1', \phi_2', \phi_3') \nonumber\\
  & = & (W \vec \phi)^T \nonumber\\
  & = & (- \phi_{1g} - 0.571 \phi_q + 0.785 \phi_{rg}, \phi_q, \phi_{rg}).
\end{eqnarray}
The flux operators of the nonfree modes are preserved, as expected. This preserves the Josephson junction term in the Hamiltonian. Also, note that since $W= (W_1^T)^{-1}$, it is its own inverse. Furthermore, the transformed $C_V$ matrix is given by
\begin{align}
  C_V' = (W^T)^{-1} C_V = 
  \begin{pmatrix}
    21.8\\
    12.5\\
    -17.2
  \end{pmatrix}.
\end{align}
Note that the voltage coupling matrix we obtain here is \textit{not} what we would obtain if we na\"ively simply remove all the free mode terms in the original Hamiltonian. In fact, had we done that, there would have been no voltage coupling after the free mode is removed.

Taking out the terms involving the free mode, Hamiltonian on the unfree modes is therefore given by
\begin{align}
  H_{\backslash F} =& \frac{1}{2} (n_2', n_3') 
  \begin{pmatrix}
    3.57 & -0.0312\\
    -0.0312 & 4.79
  \end{pmatrix}
  \begin{pmatrix}
    n_2'\\
    n_3'
  \end{pmatrix}\nonumber \\
  & - 3 \cos \phi_q + \frac{1}{2} \times 0.0204 \times \phi_{rg}^2 \nonumber\\
  & - V_v 
  \begin{pmatrix}
    12.5 & -17.2
  \end{pmatrix} 
  \begin{pmatrix}
    3.57 & -0.0312\\
    -0.0312 & 4.79
  \end{pmatrix}
  \begin{pmatrix}
    n_2'\\
    n_3'
  \end{pmatrix}.
  \label{eq:unfree_ham}
\end{align}

Now that the Hamiltonian on the modes relevant for circuit dynamics is obtained, we can conduct numerical analyses. For example, we can diagonalize the Hamiltonian to obtain the ten lowest energy eigenvalues shown in~\cref{tab:energy} and plotted in~\cref{fig:unfree_spec}. We can compute the frequency of the circuit, which is the difference between the two lowest eigenvalues: $-0.00981-(-0.999) = 0.989$GHz.
\begin{table}
  \caption{Table of energy levels for Hamiltonian in~\cref{eq:unfree_ham}} \label{tab:energy}
  \begin{ruledtabular}
  \begin{tabular}{|c|c|c|c|c|c|}
    \hline
    Energy level & 0 & 1 & 2 & 3 & 4   \\
    \hline
    Energy (GHz$\cdot h$) & -0.999 & -0.00981 & 0.979 & 1.88 & 1.97   \\
    \hline
    Energy level & 5 & 6 &  7 & 8 & 9 \\
    \hline
    Energy (GHz$\cdot h$) & 2.87 & 2.96 & 3.32 & 3.85 & 3.95 \\
    \hline
  \end{tabular}
  \end{ruledtabular}
\end{table}
\begin{figure}
  \includegraphics[width = 0.48\textwidth]{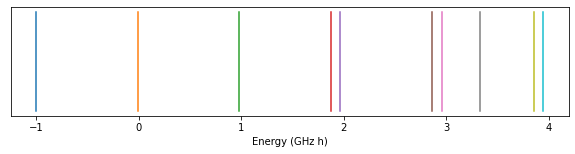}
  \caption{Spectrum of energies listed in~\cref{tab:energy}. \label{fig:unfree_spec}}
\end{figure}

\section{Mode Decoupling Techniques}\label{sec:decouple}
We now turn to the use of linear canonical transformation to decouple the different modes in a superconducting circuits. 

In general, a complex superconducting circuit has many modes that can be strongly coupled. Simulating such a many-body system can be a daunting task. However, it is well-known that the weakly coupled case can be handled via perturbation theory. For completeness we give a high level explanation here. In general, the Hamiltonian can be expressed as
\begin{align}
  H = H_\text{local} + H_\text{couple},
\end{align}
where $H_\text{local}$ consists of all local terms and $H_\text{couple}$ consists of all coupling terms. Then, when the system is weakly coupled, $H_\text{couple}$ can be treated as a perturbation. $H_\text{local}$ is first diagonalized, which simply involves diagonalizing each local Hamiltonian. This can be done using standard numerical techniques for quantum systems with one degree of freedom. For superconducting circuits, this can be done by an appropriate choice of basis depending on the type of mode being considered. For details, see for instance~\cite{rodriguez_2014, kerman2020efficient}. Then, by perturbation theory, the low energy eigenstates of $H$ approximately involve only low energy eigenstates of $H_\text{local}$. Hence, $H$ can be projected into the low energy eigenspace of $H_\text{local}$ and the low energy spectrum and eigenstates of $H$ are approximately preserved. Now, we can express $H_\text{couple}$ as a sum of tensor products of local operators, which are expressed in the eigenbasis of the corresponding local Hamiltonian. Entries involving high energy eigenstates are then truncated, which implements the low energy projection for $H_\text{couple}$. Finally, the diagonalized $H_\text{local}$ for the first few eigenstates are added. The resulting Hamiltonian is a good approximation to $H$ in the low energy eigenspace.

\begin{widetext}
We propose to simulate strongly coupled multimode circuits by reducing the problem to the weakly coupled case. The circuit Hamiltonian in~\cref{eq:circ_ham} can be split into local and coupling terms:
\begin{align}
  H = H_\text{local} + H_\text{couple},
\end{align}
where
\begin{align}
  H_\text{local} &  \equiv \sum_{i=1}^{n_J} \Bigg[\frac{1}{2} \mathcal{C}^{-1}_{ii} Q_i^2 + \frac{1}{2} (M_0)_{ii} \Phi_i^2 + (N \vec \Phi_x)_i \Phi_i - E_{J,i} \cos \left(\frac{2 \pi}{\Phi_0} \Phi_i \right)- Q_i (\mathcal{C}^{-1} C_V \vec V)_i\Bigg] \nonumber\\
  & + \sum_{i=1}^{n_L} \Bigg[\frac{1}{2} \mathcal{C}^{-1}_{ii} Q_i^2 + \frac{1}{2} (M_0)_{ii} \Phi_i^2 + (N \vec \Phi_x)_i \Phi_i - Q_i (\mathcal{C}^{-1} C_V \vec V)_i\Bigg], \label{eq:local_ham}
\end{align}
$n_J, n_L$ being respectively the number of junction and inductor modes, and
\begin{align}
  H_\text{couple} \equiv \sum_{\substack{i,j=1\\i\neq j}}^{n} \frac{1}{2} \mathcal{C}^{-1}_{ij} Q_i Q_j + \frac{1}{2} (M_0)_{ij} \Phi_i \Phi_j,
\end{align}
where $n \equiv n_J + n_L$ is the total number of modes. 
\end{widetext}
Note that we omit the term proportional to the identity coming from the quadratic term in the voltages $\vec V$ since this does not affect the physics. Now, if the coupling terms $\mathcal{C}_{ij}^{-1}, (M_0)_{ij}$ are small, then $H$ describes a weakly coupled quantum many-body system for which the above method can be applied. 

Since the coupling terms are simply quadratic forms, we turn to linear canonical transformations given in~\cref{eq:lin_can}, which transforms the flux and charge coupling matrices as~\cref{eq:coupling_transform}. By choosing appropriate transformations, we can reduce the coupling terms, which are given by the off-diagonal entries of $M_0, \mathcal{C}^{-1}$. Then,~\cref{lem:lin_can} implies that the transformed Hamiltonian will be unitarily equivalent to the original Hamiltonian. In particular, it has the same spectrum as well as the same drive (term in Hamiltonian that depends on $\vec V$) matrix elements between different eigenstates as the original Hamiltonian. Our decoupling methods can reduce the coupling terms sufficiently to the point that we are able to treat the system as weakly coupled. This makes it numerically feasible to compute both static properties of strongly coupled multimode superconducting circuits, such as spectra, and dynamical properties, such as how a voltage source drives the system. The latter can be computed via the drive matrix elements that follows from diagonalizing the transformed Hamiltonian. Furthermore, throughout we will use the sum of squares of the off-diagonal terms as a rough heuristic of how strongly coupled the circuit modes are. However, this is not directly related to the numerical cost of simulation, and in the end we will need to explicitly compute the necessary Hilbert space dimension to achieve an accurate simulation.

\subsection{Simultaneous Approximate Diagonalization}
One natural linear transformation to use is an orthogonal transformation. Then, $(W^T)^{-1} = W$. This is equivalent to simultaneously orthogonally diagonalizing $\mathcal{C}^{-1}$ and $M_0$. However, in general, the two matrices do not commute and so exact diagonalization is impossible. Instead, we can attempt to approximately simultaneously diagonalize. For instance, an optimization task can be defined: finding an orthogonal matrix $W$ such that the sum of the squares of the off-diagonal terms of the transformed matrices are minimized. However, the orthogonal transformation is restricted to the $n_L$ inductor modes so that the cosine junction terms in the Hamiltonian do not have linear combinations of fluxes, which would make them into coupling terms. This can be seen by expanding cosine as a power series.

This optimization is difficult in general, but one possible approach is to use the algorithm outlined in~\cite{cardoso1996jacobi}. The algorithm given is iterative: it cycles through different rotation axes and rotates by the optimal angle, which has an analytical solution, to minimize the sum of the squares for each axis. This is done until the rotation angles fall below a certain threshold or the maximum number of iterations is exceeded. There are no runtime or convergence guarantees given in the paper. Again, to keep the junction terms local, only the rotation axes corresponding to inductor modes are iterated through. However, the objective function still includes the off-diagonal terms involving junction modes. This technique will be referred to as the simultaneous approximate diagonalization technique.

\subsection{Inductor-only Symplectic Diagonalization}
Another method is to take advantage of a classic theorem regarding symplectic matrices. Similar to how real symmetric matrices can be diagonalized by an orthogonal matrix and normal matrices can be diagonalized by a unitary matrix, a positive definite matrix can be diagonalized by a symplectic matrix. The diagonal matrix also takes on a special form. The precise statement is given by the following. 
\begin{thm}
  \label{thm:williamson}
  (Williamson~\cite{williamson1936algebraic}). Let $M \in \R^{2n \times 2n}$ be a positive definite matrix. Then there exists a symplectic matrix $S \in \mathrm{Sp}(2n)$ and a diagonal matrix $\Lambda \in \R^{n\times n}$ with positive diagonal entries such that
  \begin{align}
    S^T M S = \mathrm{diag}(\Lambda, \Lambda).
  \end{align}
\end{thm}
\noindent The use of symplectic matrices to diagonalize quadratic bosonic Hamiltonians is well-known in literature~\cite{hormander1995symplectic,nam2016diagonalization}. This was even used to find canonical forms of superconducting circuit Hamiltonians~\cite{parra2018quantum}. In our work, we use them to simplify the numerical analysis of superconducting circuits. Furthermore, we actually prove a stronger statement of Williamson's theorem for the special case when $M$ is block diagonal. Although not difficult to prove, this statement may not be appreciated in the literature at large~\footnote{There is a related result: Lemma 103 of~\cite{de2011symplectic}. It considers the symplectic diagonalization of block diagonal positive definite matrices, but does not conclude our result. }. We now state the result:
\begin{cor}
  \label{cor:block_williamson}
  If $M \in \R^{2n \times 2n}$ is positive definite and is block diagonal with two $n\times n$ blocks, there exists a block diagonal symplectic matrix 
  \begin{align}
    S  = 
    \begin{pmatrix}
      S_n & 0_n\\
      0_n & (S_n^T)^{-1}
    \end{pmatrix} \in \mathrm{Sp}(2n)
  \end{align}
and a diagonal matrix $\Lambda \in \R^{n\times n}$ with positive diagonal entries such that
  \begin{align}
    S^T M S = \mathrm{diag}(\Lambda, \Lambda).
  \end{align}
\end{cor}
\begin{proof}
  We include the proof in~\cref{app:williamson_proof} and for purely implementation purposes, see
supplemental material for the code for finding the desired block diagonal symplectic matrix. 
\end{proof}
\noindent \cref{cor:block_williamson} allows us to diagonalize the quadratic part of the circuit Hamiltonian via a linear canonical transformation as defined in~\cref{eq:lin_can}. In particular, with our block diagonal formulation, flux and charge operators are not mixed with each other. One advantage of this is physical interpretability: the Hamiltonian in terms of the transformed operators is still of the form of a superconducting circuit Hamiltonian. For instance, unphysical terms such as mixed coupling terms involving a flux operator and a charge operator do not appear~\footnote{However, these terms could appear for more general circuits with linear non-reciprocal elements, as seen in~\cite{parra2018quantum}. }. This allows us to use established intuition about circuit Hamiltonians to guide design and possible gate schemes. This also allows us to compose different techniques for analyzing superconducting circuit Hamiltonians. 

Now, assume $M_0$ is positive definite. Considering the definition of $M_0$ in~\cite{burkard2005circuit}, this is true for instance if every Josephson junction is shunted by an inductor, that is, we are working with fluxonium qubits. We could also artificially enforce this by shunting all junctions with inductors with very large inductances~\footnote{This would lead to different boundary conditions for the junction mode as shown in~\cite{koch2009charging}. However, as pointed out in~\cite{vool2017engineering}, this should not affect what we see in experiments due to the large timescales that would have to be involved. }. Then, the block matrix
\begin{align}
  \mathcal Q \equiv
  \begin{pmatrix}
    M_0 & 0_n \\
    0_n & \mathcal{C}^{-1}
  \end{pmatrix}
\end{align}
is positive definite. This matrix encodes the quadratic part of the Hamiltonian, including all the coupling terms. Hence, it is immediate that we should find a linear canonical transformation that diagonalizes $\mathcal{Q}$ by invoking~\cref{cor:block_williamson}, thereby eliminating all the coupling.

However, in order to keep the junction terms in the Hamiltonian local, only the inductor modes are to be transformed. Hence consider the block matrix
\begin{align}
  \mathcal{Q}_L \equiv
  \begin{pmatrix}
    (M_0)_L & 0_{n_L} \\
    0_{n_L} & (\mathcal{C}^{-1})_L
  \end{pmatrix},
\end{align}
where $_L$ means only the submatrices corresponding to inductor modes are taken. Since $\mathcal{Q}$ is positive definite, so is $\mathcal{Q}_L$. Using~\cref{cor:block_williamson}, the matrix $\mathcal{Q}_L$ can be diagonalized by a block diagonal symplectic matrix 
\begin{align}
  S_L = 
  \begin{pmatrix}
    S_{n_L} & 0_{n_L}  \\
    0_{n_L} & (S_{n_L}^T)^{-1}
  \end{pmatrix}
\end{align}
such that $S_L^T \mathcal{Q}_L S_L = \mathrm{diag}(\Lambda, \Lambda)$. Then, by linearly transforming the modes with
\begin{align}
  W \equiv
  \begin{pmatrix}
    I_{n_J} & 0_{n_J \times n_L} \\
    0_{n_L \times n_J} & S_{n_L}^{-1}
  \end{pmatrix}.
\end{align}
the coupling matrices transform as
\begin{align}
  M_0 \mapsto (W^T)^{-1} M_0 W^{-1}=
  \begin{pmatrix}
    (M_0)_J & (M_0)_{LJ}^T S_{n_L}  \\
    S_{n_L}^T (M_0)_{LJ} & \Lambda
  \end{pmatrix},
\end{align}
\begin{align}
  \mathcal{C}^{-1} \mapsto W \mathcal{C}^{-1} W^T=
  \begin{pmatrix}
    \mathcal (C^{-1})_J & (\mathcal C^{-1})_{LJ}^T (S_{n_L}^T)^{-1}\\
    S_{n_L}^{-1} (\mathcal C^{-1})_{LJ} & \Lambda
  \end{pmatrix},
\end{align}
where $_J$ and $_{JL}$ correspond to the submatrices for the junction modes and the coupling coefficients between junction and inductor modes, respectively. Recalling that $\Lambda$ is a diagonal matrix, this effectively fully decouples the inductor modes from each other. However, there might still be coupling between inductor modes and junction modes. The coupling between junction modes is left intact. This will be referred to as the inductor-only symplectic diagonalization decoupling technique.

\subsection{Full Symplectic Diagonalization}\label{sec:full_symplectic}
The final method proposed is to fully diagonalize the quadratic part of the Hamiltonian. This will introduce coupling terms via the Josephson junction terms in the Hamiltonian, but in some cases these new couplings will not be very large. We will also lose the physical interpretability that the previous techniques had.

Specifically, by~\cref{cor:block_williamson}, there exists a symplectic matrix
\begin{align}
  S =
  \begin{pmatrix}
    S_n & 0_n \\
    0_n & (S_n^T)^{-1}
  \end{pmatrix}
\end{align}
such that $S^T \mathcal{Q} S = \mathrm{diag}(\Lambda, \Lambda)$, assuming $M_0$ is positive definite. Then, transform the modes via the linear transformation
\begin{align}
  W \equiv S_n^{-1}.
\end{align}
This completely diagonalizes the quadratic part of the Hamiltonian, thereby removing all quadratic coupling terms. However, the transformed modes are still coupled through the transformed junction term:
\begin{align}
  \sum_{i=1}^{n_J} E_{J,i} \cos \left(\frac{2 \pi}{\Phi_0} \Phi_i \right)  \mapsto \sum_{i=1}^{n_J} E_{J,i} \cos \left(\frac{2 \pi}{\Phi_0} \sum_{j=1}^{n} (S_n)_{ij}\Phi_j' \right),
\end{align}
where $\vec \Phi' \equiv W \vec \Phi$ are the transformed flux variables. However, the power series of cosine leads to additional local terms. The local and coupling terms can be separated as follows:
\begin{align}
  & \sum_{i=1}^{n_J}E_{J,i} \cos \left(\frac{2 \pi}{\Phi_0} \sum_{j=1}^{n} (S_n)_{ij}  \Phi_j' \right) \nonumber\\
  & = \sum_{i=1}^{n_J}E_{J,i} \sum_{k=0}^\infty \frac{(-1)^{k}}{2k!} \left(\frac{2 \pi}{\Phi_0} \sum_{j=1}^{n} (S_n)_{ij} \Phi_j' \right)^{2k}\nonumber\\
  & = \sum_{i=1}^{n_J}E_{J,i} \sum_{k=0}^\infty \frac{(-1)^k}{2k!} \left[\sum_{j=1}^{n} \left(\frac{2 \pi}{\Phi_0} (S_n)_{ij} \Phi_j' \right)^{2k} \right] - H_\text{couple}' \nonumber\\
  & = \sum_{j=1}^{n} \sum_{i=1}^{n_J}E_{J,i} \cos \left(\frac{2 \pi}{\Phi_0} (S_n)_{ij} \Phi_j' \right) - H_\text{couple}'.
\end{align}
\begin{widetext}
Hence, in summary the transformed Hamiltonian is given by
\begin{align}
  H =  H_\text{local}'+  H_\text{couple}' ,
\end{align}
where
\begin{align}
& H_\text{local}' \equiv \sum_{j=1}^{n} \Bigg[ \frac{1}{2}\Lambda_j (\Phi_j'^2 + Q_j'^2) + ( (W^T)^{-1} N) \vec \Phi_x)_j  \Phi_j' - \sum_{i=1}^{n_J}E_{J,i} \cos \left(\frac{2 \pi}{\Phi_0} (S_n)_{ij} \Phi_j' \right) - ( (W^T)^{-1} C_V \vec V)_j \Lambda_j Q_j'\Bigg],
\end{align}
and
\begin{align}
  \label{eq:full_symp_couple}
  H_\text{couple}' = - \sum_{i=1}^{n_J} E_{J,i} \Bigg[ \cos \left(\frac{2 \pi}{\Phi_0} \sum_{j=1}^{n} (S_n)_{ij} \Phi_j'  \right) - \sum_{j=1}^{n}  \cos \left(\frac{2 \pi}{\Phi_0} (S_n)_{ij} \Phi_j' \right) \Bigg].
\end{align}
\end{widetext}
Since the transformation mixes inductor and junction modes, unlike~\cref{eq:local_ham}, the distinction between the two types of modes disappear. This will be referred to as the full symplectic diagonalization decoupling technique. 

For each of the above decoupling techniques, an optimization over spanning trees can be done. This simply iterates over all spanning trees of the circuit consisting of all junctions, all voltage sources, and inductors. Different spanning trees lead to different magnitudes of the coupling terms in the circuit Hamiltonian. A possible optimization that can be defined is to find which spanning tree, combined with either the simultaneous approximate diagonalization or the inductor-only symplectic diagonalization decoupling techniques, leads to transformed $M_0, \mathcal{C}^{-1}$ matrices with the smallest sum of the squares of the off-diagonal entries. A possible optimization for the full symplectic diagonalization decoupling technique is to minimize sum of the squares of the $n_J$ rows of $S_n$ that correspond to the junction modes.

\subsection{Example}
For demonstration purposes, we apply these techniques to the circuit of two inductively coupled fluxonium qubits shown in \cref{fig:ind_coupled}. For clarity, capacitors on edges [1,3], [1,4], [2,3], [2,4], [1,5], [4,6] are not explicitly shown in this figure. The capacitances and inductances are given in~\cref{tab:cap} and~\cref{tab:ind}, respectively.

\begin{figure}
  \centering
  \includegraphics[width = 0.48 \textwidth]{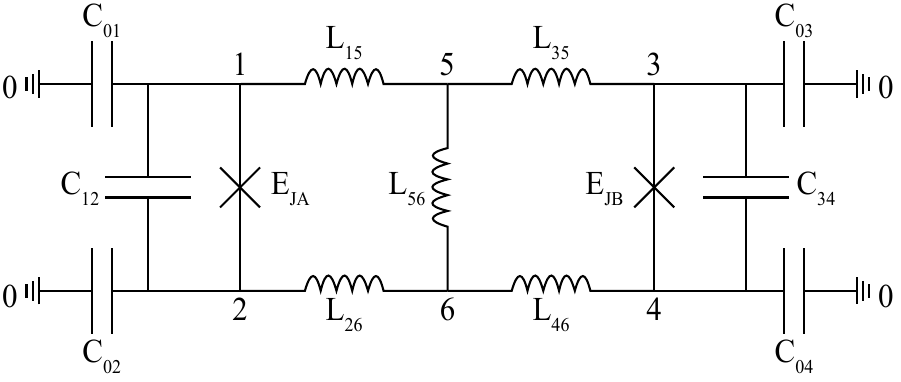}
  \caption{Circuit diagram of two inductively coupled fluxonium qubits. The spanning tree edges are $[1,2], [3,4], [0,1], [1,5], [4,6], [5,6]$. \label{fig:ind_coupled}}
  
\end{figure}

With these parameters, and after removing the free modes as in~\cref{sec:free_mode}, we obtain the Hamiltonian 
\begin{align}
  H = & \frac{1}{2} \vec n^T  \mathcal{C}^{-1} \vec n - \sum_{i=1}^{n_J} E_{J,i} \cos \left(\phi_i \right) + \frac{1}{2} \vec \phi^T M_0 \vec \phi 
\end{align}
with coupling matrices
\begin{align}
	\mathcal{C}^{-1}=
	\begin{pmatrix}
	13.6   & 0.276    & -6.85    & 0.182 & -0.110\\
	0.276  & 12.3     & -0.00711 & 6.027 & -0.0999\\
	-6.85  & -0.00711 & 400      & 378   & -12.2\\
	0.182  & 6.027    & 378      & 400   & 12.2\\
	-0.110 & -0.0999  & -12.2    & 12.2  & 24.6
	\end{pmatrix},
\end{align}

\begin{align}
	M_{0}=
	\begin{pmatrix}
	1.00 & 0     & 1.00 & 0     & 1.00\\
	0    & 1.00  & 0    & -1.00 & 1.00\\
	1.00 & 0     & 2.00 & 0     & 1.00\\
	0    & -1.00 & 0    & 2.00  & -1.00\\
	1.00 & 1.00  & 1.00 & -1.00 & 27.1
	\end{pmatrix}.
\end{align}
This is a quantum system with 6 interacting degrees of freedom. The large number of modes follows directly from the structure of the circuit, which involves an inductor bridge on edge $[5,6]$ to provide the inductive coupling. The coupling terms in the Hamiltonian correspond to the off-diagonal terms in the coupling matrices. The sum of the squares of the off-diagonal terms is $2.87 \times 10^5$. Note that throughout we will keep 3 significant figures and use GHz$\cdot h$ as units. Furthermore, we will set $E_{JA} = 3.9 \mathrm{GHz} \cdot h, E_{JB} = 4.4 \mathrm{GHz} \cdot h$. In the following calculations, we will flip the sign of the Josephson junction terms. This can be done via an external magnetic flux~\footnote{Note that in~\cref{eq:circ_ham}, the external magnetic flux comes in as a linear functional of the flux operators. By an appropriate transformation of the flux operators, this can be turned into a phase shift in the junction terms, assuming the external flux is time independent. This is studied in detail for instance in~\cite{you2019circuit}. } and is the well-known ``sweet spot'' for fluxonium qubits at which there is a first-order insensitivity to flux due to symmetry. This is the point of operation for experiments (see for instance~\cite{nguyen2019high}).

\begin{table}
  \caption{Capacitance table. Every edge is defined by the pair of nodes with labels defined in~\cref{fig:ind_coupled}. The corresponding capacitance is extracted from a realistic layout design and the parasitic capacitances of the Josephson junctions. The capacitances on the edges not listed in the table are neglected. }
  \begin{ruledtabular}
    \begin{tabular}{|c|c|c|c|c|c|c|c|c|}
      \hline 
      Edge &[0,1]&[0,2]&[0,3]&[0,4]&[1,2]&[3,4]&[1,3]&[1,4]\\
      \hline
      Capacitance (fF) &8.3&9.3&10.6&11.8&6.4&6.4&0.9&0.3\\
      \hline
      Edge &[2,3]&[2,4]&[1,5]&[2,6]&[3,5]&[4,6]&[5,6]&\\
      \hline
      Capacitance (fF) &0.3&0.7&0.1& 0.1&0.1&0.1&6.2&\\
      \hline
    \end{tabular}
  \end{ruledtabular}
  \label{tab:cap}
\end{table}

\begin{table}
  \caption{Inductance table. Every edge is defined by the pair of nodes with labels defined in~\cref{fig:ind_coupled}.}
  \begin{ruledtabular}
     \begin{tabular}{|c|c|c|c|c|c|}
      \hline 
      Edge & [1,5]&[2,6]&[3,5]&[4,6]&[5,6]\\
      \hline
      Inductance (nH) & 163.5& 163.5& 163.5& 163.5& 6.5\\
      \hline
    \end{tabular}
  \end{ruledtabular}
  \label{tab:ind}
\end{table}

We can directly diagonalize the Hamiltonian and plot the energies of the low energy states as a function of the local basis cutoff $d$. That is, we project the overall Hamiltonian onto the tensor product of the $d$-dimensional lowest energy eigenspaces of the local Hamiltonians. The results are shown in~\cref{fig:E-cutoff}(a). The dimension $d$ is for each individual mode, so the overall dimension is $d^5$.

We apply the simultaneous approximate diagonalization technique, with which we obtain the orthogonal transformation
\begin{align}
  W = 
  \begin{pmatrix}
    1&          0&          0&          0&          0 \\
    0&          1&          0&          0&          0 \\
    0&          0&          0.707&      0.707&      0 \\
    0&          0&         -0.510&      0.510&     -0.693\\
    0&          0&         -0.490&      0.490&      0.721 
  \end{pmatrix}.
\end{align}
Note that the first two modes are junction modes and therefore do not transform. This yields the transformed coupling matrices
\begin{align}
  & W \mathcal{C}^{-1} W^T \nonumber\\ 
  & =  
  	\begin{pmatrix}
	13.6   & 0.276    & -4.72    & 3.66   & 3.37\\
	0.276  & 12.3     & 4.26     & 3.15   & 2.89\\
	-4.72  & 4.26     & 778      & 0      & 0 \\
	3.66   & 3.15     & 0        & 5.89   & -0.667\\
	3.37   & 2.89     & 0        & -0.667 & 40.6
	\end{pmatrix}
\end{align}
and
\begin{align}
  & (W^T)^{-1} M_0 W^{-1}  \nonumber\\ 
  & = W M_0 W^T  \nonumber\\ 
  & =
  	\begin{pmatrix}
	1.00  & 0     & 0.707 & -1.20  & 0.231 \\
	0     & 1.00  &-0.707 & -1.20  & 0.231 \\
	0.707 &-0.707 & 2.00  & 0      & 0\\
	-1.20 & -1.20 & 0     & 15.5   & -12.6\\
	0.231 & 0.231 & 0     &-12.6   & 13.7
      \end{pmatrix}.
\end{align}
With this transformation, the sum of the squares of the off-diagonal terms is reduced to 494, a reduction of three orders of magnitude. We plot the state energies as a function of local basis dimension in~\cref{fig:E-cutoff}(b), which shows faster convergence than the case where no decoupling techniques is applied as in~\cref{fig:E-cutoff}(a).

Next we apply the inductor-only symplectic diagonalization decoupling technique. Using the above mentioned algorithm we obtain the transformation
\begin{align}
  W =
  	\begin{pmatrix}
    1&          0&          0&          0&          0 \\
    0&          1&          0&          0&          0 \\
    0&          0&          0.476&     -0.477&      0.474 \\
    0&          0&          0&          0&          1.02\\
    0&          0&         -0.159&     -0.159&      0 
  	\end{pmatrix}.
\end{align}
Again, only the inductor modes are transformed. The transformed coupling matrices are 
\begin{align}
  W \mathcal{C}^{-1} W^T =
    \begin{pmatrix}
     13.6&   0.276& -3.40&  -0.113&  1.06 \\
	 0.276&  12.3&  -2.92&  -0.103& -0.958\\
	-3.40&  -2.92&   4.40&   0&      0    \\
	-0.113& -0.103&  0&      25.4&   0    \\
	 1.06&  -0.958&  0&      0&      39.4 
	\end{pmatrix}
\end{align}
and
\begin{align}
  & (W^T)^{-1} M_0 W^{-1} \nonumber\\
  & = 
  	\begin{pmatrix} 
  	1.00&    0&      1.05&  0.494&  -3.14& \\
	0&       1.00&   1.05&  0.495&   3.14& \\
	1.05&    1.05&   4.40&  0&        0   \\
	0.494&   0.495&  0&     25.4&     0    \\
	-3.14&   3.14&   0&     0&       39.4
      \end{pmatrix}.
\end{align}
As expected, the submatrix for the inductor subspace is diagonal. The sum of the squares of off-diagonal terms is 89.3, which is less than one fifth of that of the simultaneous approximate diagonalization technique. Also, the largest elements are about one fourth as large. We plot the state energies as a function of local basis dimension for this technique in~\cref{fig:E-cutoff}(c), which shows faster convergence than those of the simultaneous approximate diagonalization and no decoupling cases.

Finally, we implement the full symplectic diagonalization technique. With the above algorithm, we obtain the transformation matrix 
\begin{align}
  & W  \nonumber\\
  	&= \begin{pmatrix}
	-0.0195&   0.441&  -0.0112&   0.0112&  -0.0137& \\
 	 0.442&    0.0122&  0.0140&  -0.0140&   0.0112& \\
 	 0.256&    0.271&   0.529&   -0.529&    0.529&  \\
 	 0.0195&   0.0195&  0&         0&       1.02 \\
 	 0.0797&  -0.0797&  0.159&    0.159&    0    
  	\end{pmatrix}.
\end{align}
The junction modes are also transformed to fully diagonalize the coupling matrices. The transformed matrices are the same diagonal matrix
\begin{align}
  W \mathcal{C}^{-1} W^T & = (W^T)^{-1} M_0 W^{-1}\nonumber \\
  & = \Lambda \nonumber \\
  & = 
  	\begin{pmatrix}
    2.46 & 0    & 0    & 0    & 0\\
    0    & 2.58 & 0    & 0    & 0\\
    0    & 0    & 3.57 & 0    & 0\\
    0    & 0    & 0    & 25.4 & 0\\
    0    & 0    & 0    & 0    & 39.4
  \end{pmatrix}.
\end{align}

Now, all the coupling terms are in the junction terms. For that we need to consider the first $n_J=2$ rows of $S_n = W^{-1}$, 
\begin{align}
  	\begin{pmatrix}
 	-0.0259&  2.29&   -0.0614&  0.00628&  0& \\
 	 2.24&     0.0721&  0.0454&  0.00569&  0& 
       \end{pmatrix},
\end{align}
whose elements appear in the transformed Hamiltonian. We plot the state energies as a function of local basis dimensions in~\cref{fig:E-cutoff}(d), which shows the fastest convergence among all the different decoupling techniques. 

\begin{figure}
  \centering
  \includegraphics[width = 0.48 \textwidth]{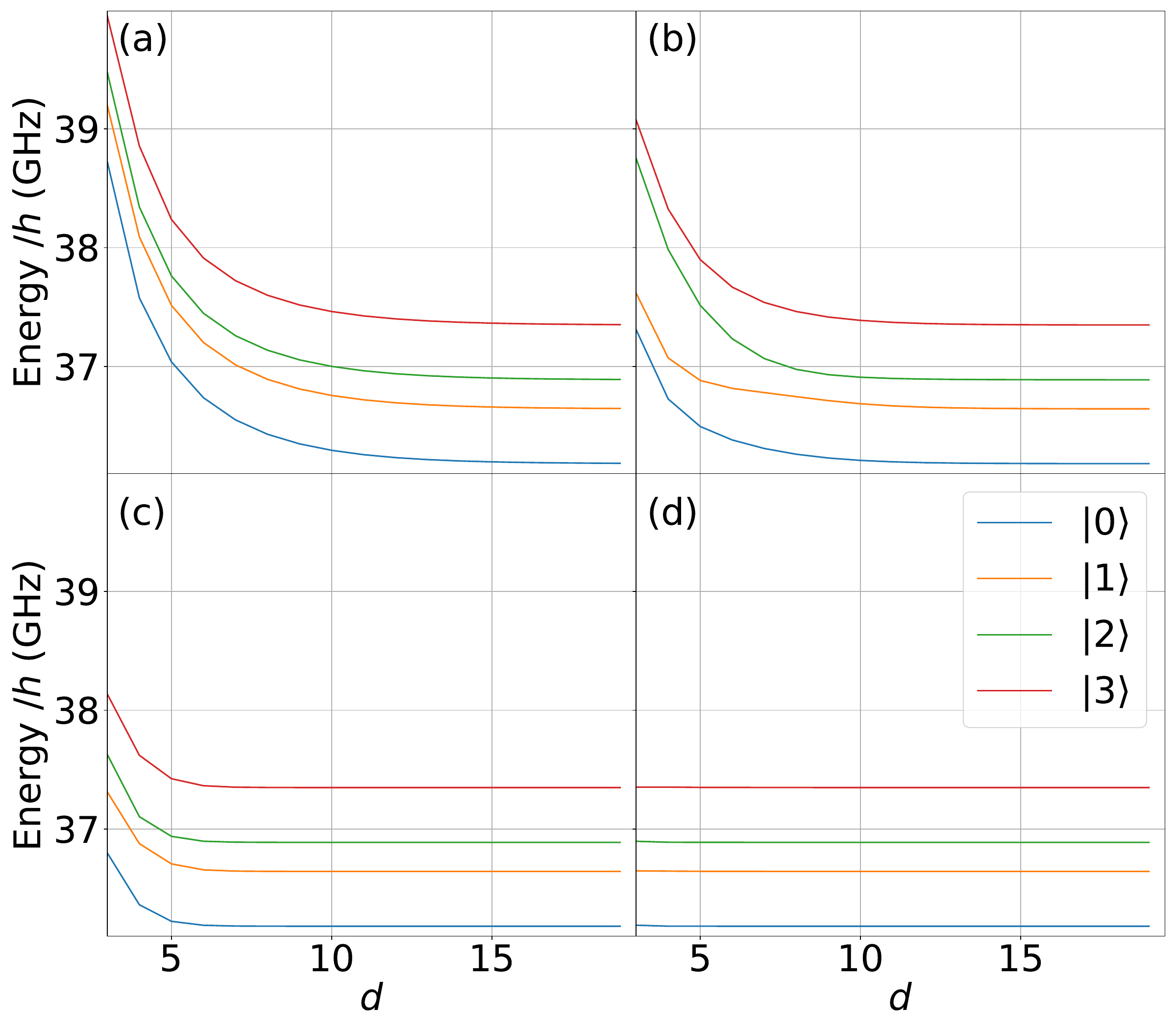}
  \caption{Energies of four lowest energy states $\left|0\right\rangle, \left|1\right\rangle, \left|2\right\rangle, \left|3\right\rangle$ as a function of local basis cutoff $d$ with (a) no decoupling, (b) simultaneous approximate diagonalization decoupling, (c) inductor-only symplectic decoupling, and (d) full symplectic diagonalization decoupling. \label{fig:E-cutoff}}
\end{figure}

To further investigate the local basis cutoff convergence, we compute the reduced density matrices $\rho_k$ for every mode $k$. In general, the eigenstate of an $n$-mode Hamiltonian has a wavefunction
\begin{align}
  \label{eq:wave_func}
  \left|\Psi\right\rangle =\sum_{\left\{ m_{k}\right\} }C\left(m_{1},m_{2},\cdots,m_{n}\right)\left|m_{1}m_{2}\cdots m_{n}\right\rangle ,
\end{align}
where $m_k$ labels the local eigenstate of mode $k$, and $C\left(m_{1},m_{2},\cdots,m_{n}\right)$ is the amplitude of state $\left|m_{1}m_{2}\cdots m_{n}\right\rangle$. The sum over $\left\{ m_{k}\right\}$ is over all possible local state configurations for the given local basis cutoffs. The reduced density matrix of mode $k$ is then obtained by tracing out all possible local states except mode $k$:
\begin{align}
  \label{eq:reduced_dense_mat}
  \rho_{k}\left(m_{k},m_{k}^{\prime}\right)=\sum_{\{m_{i}\neq m_{k}\}}& C^{*}\left(m_{1},\cdots,m_{k},\cdots,m_{n}\right)\nonumber\\
  & C\left(m_{1},\cdots,m_{k}^{\prime},\cdots,m_{n}\right).
\end{align}
The diagonal entry $\rho_k\left(m_k, m_k\right)$ represents the population of local state $m_k$ of mode $k$. If mode $k$ is well decoupled from other modes, $\rho_k\left(m_k, m_k\right)$ should decay quickly with increasing $m_k$ for low energy states of the overall Hamiltonian. Therefore, we can compare the performance of the decoupling techniques by computing $\rho_k\left(m_k, m_k\right)$ as a function of $m_k$ for all modes, as shown in~\cref{fig:weight-index}. From this figure, we can see that for all three decoupling techniques, the local state populations decay faster than when no decoupling technique applied. The techniques in increasing order of efficacy are the simultaneous approximate diagonalization technique, the inductor-only symplectic diagonalization technique, and the full symplectic diagonalization technique. This is consistent with what we see in~\cref{fig:E-cutoff}. Note that although we have presented the different decoupling techniques in increasing efficacy for this example, there is a priori no reason to expect this for other circuits.

Figure \ref{fig:weight-index} also shows that the decay rate of the local state populations may be significantly different. For instance in~\cref{fig:weight-index}(b), when $m_k=5$, the population of mode $4$ is about four orders of magnitude larger than that of mode $3$. This inspires an algorithm to find an appropriate local basis cutoff for each mode. We can choose a population threshold $\epsilon$ and dynamically determine the minimal basis cutoff $d_k$ of each mode to satisfy the criterion $\rho_k\left(d_k,d_k\right) < \epsilon$. As a result, the cutoff of each mode may be very different, and the computational cost can by greatly reduced compared to using a uniform cutoff to achieve the same accuracy. For instance, if we set $\epsilon=10^{-17}$ and apply the full symplectic diagonalization decoupling technique to compute the ground state, the cutoff dimensions are $(54, 53, 8, 5, 2)$. The total Hilbert space dimension is more than three orders of magnitude smaller than that of using a uniform cutoff $(54, 54, 54, 54, 54)$, and the computational cost can be about six orders of magnitude smaller if we use the standard Arnoldi iteration method to diagonalize the Hamiltonian.
 
\begin{figure}
  \centering
  \includegraphics[width = 0.48 \textwidth]{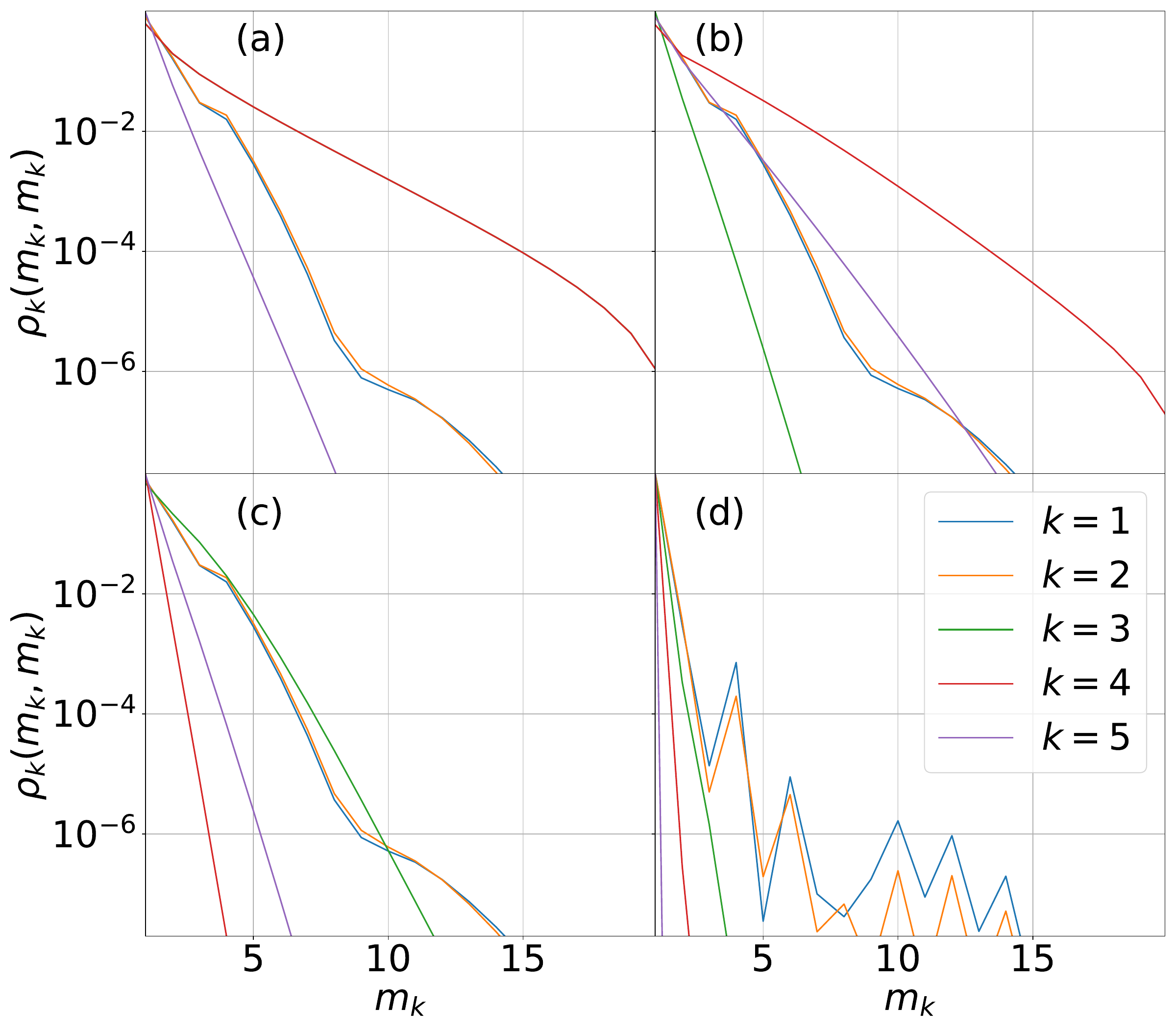}
  \caption{The diagonal entries of ground state reduced density matrices for each mode as a function of local state $m_k$ with (a) no decoupling, (b) simultaneous approximate diagonalization decoupling, (c) inductor-only symplectic decoupling, and (d) full symplectic diagonalization decoupling. In panel (a), the lines of $k=3$ (blue) and $k=4$ (red) are almost on top of each other. \label{fig:weight-index}}
\end{figure}

Using this algorithm to find the optimal local basis cutoffs, we can directly compare our different decoupling techniques via the Hilbert space sizes needed to go below the same population threshold $\epsilon$. For comparison and so that the runtimes are reasonable, we set $\epsilon = 10^{-5}$. The results are shown in~\cref{tab:hilbert_size}. The total Hilbert space size is reduced from 98,304 to 360, a two to three orders of magnitude reduction. Another interesting observation is that since the full symplectic decoupling technique mixes both inductor and junction modes, the local basis cutoffs look qualitatively different from those of the other techniques.
\begin{table}
	\caption{Different decoupling techniques and the local basis cutoffs needed to achieve $\epsilon =10^{-5}$. For succinctness we use the acronyms SAD for simultaneous approximate diagonalization, IS for inductor-only symplectic, and FS for full symplectic.}
  \begin{ruledtabular}
     \begin{tabular}{|c|c|c|}
      \hline 
      Technique & Local basis cutoffs & Total dimension \\
      \hline
      None & (8, 8, 16, 16, 6) & 98,304\\
      \hline
      SAD & (8, 8, 5, 16, 10) & 51,200\\
      \hline
      IS & (8, 8, 9, 4, 5) & 11,520\\
      \hline
      FS & (6, 5, 3, 2, 2) & 360\\
      \hline
    \end{tabular}
  \end{ruledtabular}
  \label{tab:hilbert_size}
\end{table}

\section{Conclusion and Future Work}\label{sec:conclusion}
In this paper we solve two general issues that arise when simulating a general superconducting quantum circuit. The first is the existence of free modes, and the second is the presence of many circuit modes. Our techniques help open the door to the simulation of complex, highly anharmonic superconducting circuits as alternative and possibly more promising paths to a scalable and universal quantum computer. The core mathematical tool we use is the linear canonical transformation, which we find suitable for the mostly quadratic circuit Hamiltonian.

However, since our Hamiltonian is not completely quadratic, it may be worthwhile to look into nonlinear canonical transformations, whose role in quantum mechanics has been studied in works such as~\cite{mello1975nonlinear,anderson1994canonical}. Another interesting direction is the wealth of literature regarding equivalences of classical circuits. This is actually one of the main ideas behind blackbox quantization for weakly anharmonic circuits, where the notion of Foster equivalent circuits is used extensively. It may be interesting to consider the quantum analog of other circuit equivalences, including nonlinear circuit equivalences. This may also be a tantalizing avenue to developing a fundamental theory of quantum circuits.

Some of our decoupling techniques can be improved in various ways. For instance, the spanning tree optimization we propose na\"ively runs through all possible spanning trees of the circuit, which, by Cayley's formula, can be up to $v^{v-2}$ for a complete graph, where $v$ is the number of nodes. Hence, it is worthwhile to study how exactly the coupling matrices in the Hamiltonian depends on the spanning tree chosen and improve this optimization procedure. Furthermore, since most of our decoupling techniques keep the superconducting circuit form of the Hamiltonian, it may be interesting to find ways to compose the different decoupling techniques to further reduce couplings. 

Another possible improvement is that for the symplectic decoupling techniques, the symplectic matrix $S$ which diagonalizes $\mathcal{Q}$ in~\cref{cor:block_williamson} is not unique. A trivial example is multiplying a diagonalizing symplectic $S$ by a two copies of a permutation matrix $P$
\begin{align}
  S \mapsto S 
  \begin{pmatrix}
    P & 0_n\\
    0_n & P
  \end{pmatrix},
\end{align}
which effectively re-orders the elements in $\Lambda$. An appropriate choice of $P$ may reduce the coupling terms in the Hamiltonian. We can even consider nonblock diagonal $S$, as it is not clear that block diagonal $S$ minimize the coupling, although it does provide physical interpretability. Also, for the symplectic decoupling techniques, we assumed that $M_0$ is positive definite, which is not true for all circuits. Although it is possible to artificially enforce this by adding large shunting inductors, we can ask for a more elegant solution by extending Williamson's theorem to the positive semidefinite case. This statement was claimed in~\cite{idel2017perturbation} but not proven, and according to private communication, it may actually still be open whether or not it holds. 

For the full symplectic diagonalization technique, we could take one step further by first Taylor expanding the cosine junction term and adding the quadratic terms to $M_0$. That is, we map
\begin{align}
  M_0 \mapsto M_0 + \left( \frac{2\pi}{\Phi_0} \right)^2
  \begin{pmatrix}
    M_{E_J} & 0_{n_J \times n_L} \\
    0_{n_L \times n_J} & 0_{n_L}
  \end{pmatrix},
\end{align}
where $M_{E_J}$ is the $n_J \times n_J$ diagonal matrix with the junction energies on the diagonal. If $M_0$ is positive definite, since $E_{J,i} >0$ the mapping would preserve positive definiteness, and we can then apply full symplectic diagonalization with the new coupling matrices. This would in effect eliminate the quadratic coupling terms in~\cref{eq:full_symp_couple}. However, the flux of the local systems might not be localized around zero, and so we may not be able to treat the low order terms in the cosine expansion perturbatively. Hence, eliminating the quadratic terms may not be what we want. This could be amended by expanding cosine around wherever the flux is localized for that local system, and we leave this for future work.

After removing free modes and reducing intermode couplings, we can use exact diagonalization (ED) to compute the eigenstates. This allows us to obtain all the relevant energy levels to simulate time evolution via matrix exponentiation, but this is feasible only for small circuits. For strongly coupled multimode circuits that are intractable for ED, DMRG can be applied, a technique that has already been demonstrated for simulating arrays of junctions~\cite{junction_chain1997,junction_ladder2003,current_mirror_DMRG2019,fluxonium_dmrg2019}. However, DMRG can efficiently compute only low-lying eigenstates, which may be insufficient for direct gate simulation. It would be interesting to explore an embedding strategy to simulate time evolution via an effective Hamiltonian that accurately describes the part of the circuit corresponding to the qubits for the gate operation as well as neighboring circuit components, while treating the rest of the circuit in a mean-field manner.

Finally, it would be interesting to study how some of our techniques can be applied to more general circuits with linear nonreciprocal elements such as gyrators and circulators~\cite{parra2019canonical}, as well as circuits coupled to infinite dimensional systems such as transmission lines~\cite{parra2018quantum}. The free particle problem has been discussed at length in both of these scenarios.

\textit{Note added}: Recently a related work~\cite{kerman2020efficient} has come to our attention. In this paper, the author also considers general superconducting circuits and develops two techniques to simulate them more efficiently. However, the question they primarily answer is: given a circuit Hamiltonian, what is an efficient way to numerically analyze it? This could be diagonalization or exponentiation for instance. In contrast, our focus is on how to transform the modes to manipulate the Hamiltonian into a form more amenable to numerical analysis. The first technique they give is a classification of modes based on their local Hamiltonian, choosing a convenient basis for each so as to make matrix representation of the relevant local operators sparse. The second is to partition the set of modes so that they can treat the coupling between different subsets as a perturbation and diagonalize each subset of modes separately. They consider separating the fast and slow modes, such as in~\cite{rodriguez_2014}. Since we focus on a different step in the simulation of a superconducting circuit, our techniques could be combined with those of~\cite{kerman2020efficient}. In our paper we additionally prove that free modes, which are referred to as island modes in~\cite{kerman2020efficient}, are fully independent and can be removed from the Hamiltonian after an appropriate transformation. 

\begin{acknowledgments}
	We thank Daniel Weiss for noticing that the eigenenergies we compute should monotonically decrease with respect to the local basis cutoff dimensions, which helped us identify a minor bug in our code~\footnote{For terms like $\cos(\phi_i)$ in the local Hamiltonians, we should compute $\cos(\phi_i)$ with the unprojected $\phi_i$ operator first, and then project onto the local low energy subspace. }. We thank Guido Burkard, Jianxin Chen, Jay Gambetta, Martin Idel, Jens Koch, Adrian Parra-Rodriguez, Mario Szegedy , Feng Wu, Xinyuan You, and Gengyan Zhang for inspiring discussions. DD would like to thank God for all of His provisions.
\end{acknowledgments}

\appendix
\section{Details of Free Mode Removal Algorithm}
\label{sec:gauss_elim}
We here go into the details of the algorithm for free mode removal. We first perform the orthogonal transformation on the modes which diagonalizes the subspace in~\cref{eq:num_free}. Next, we define the linear canonical transformation matrix $W$ that implements the Gaussian elimination. For $f\in \{1,2,\cdots,F\}$, we iteratively define the matrices $W_f$ and $\mathcal{C}_f$. $W_f$ is the $n \times n$ identity matrix except for its $f$-th column, which has the following entries:
\begin{align}
  \label{eq:W_def}
  (W_f)_{if} \equiv - \frac{(\mathcal{C}_{f-1})_{if}}{(\mathcal{C}_{f-1})_{ff}}.
\end{align}
Next, $\mathcal{C}_0$ is defined to be $\mathcal{C}$, while
\begin{align*}
  \mathcal{C}_f \equiv W_f \mathcal{C}_{f-1} W^T_f.
\end{align*}
Our overall transformation matrix is then given by
\begin{align}
  \label{eq:free_W}
  W \equiv [(W_F W_{F-1} \cdots W_1)^T]^{-1}.
\end{align}

$W_f$ is well-defined since $\mathcal{C}_{f-1}$ is positive definite and so $(\mathcal{C}_{f-1})_{ff}$ is never zero. This can be established by induction. $\mathcal{C}_0 \equiv \mathcal{C}$ is positive definite as defined in~\cite{burkard2005circuit}. Now assume the same for $\mathcal{C}_{f-1}$. This implies $W_f$ is well-defined. Since $(W_f)_{ff} = -1 \neq 0$, the $f$-th column of $W_f$ is linearly independent of the other columns, which are those of the identity matrix. Thus, $W_f$ has full rank. This implies $\mathcal{C}_f \equiv W_f \mathcal{C}_{f-1} W^T_f$ is positive definite. The claim follows.

The final matrix
\begin{align}
  \label{eq:gauss_elim}
  \mathcal{C}' \equiv \mathcal{C}_F = (W^T)^{-1} \mathcal{C} W^{-1}
\end{align}
has vanishing off-diagonal components for its first $F$ rows and columns. This can be verified via the following. The off-diagonal entries of $f$-th column of $\mathcal{C}_f$ are given by:
\begin{align*}
  (\mathcal{C}_f)_{if} & = \sum_{j,k=1}^{n} (W_{f})_{ij} (\mathcal{C}_{f-1})_{jk} (W_{f})_{fk}\nonumber\\
  & = - \sum_{j=1}^{n} (W_{f})_{ij} (\mathcal{C}_{f-1})_{jf} \nonumber\\
  & = - \left[- \frac{(\mathcal{C}_{f-1})_{if}}{(\mathcal{C}_{f-1})_{ff}} (\mathcal{C}_{f-1})_{ff} + 1 \times (\mathcal{C}_{f-1})_{if} \right]\nonumber\\
  & = 0.
\end{align*}
By the symmetry of $\mathcal{C}_f$, the off-diagonal entries in the $f$-th row is also vanishing. Furthermore, it can be shown that the off-diagonal terms of $\mathcal{C}_f$ for rows and columns before $f$ are also vanishing. This can be established via induction. This is vacuously true for $\mathcal{C}_1$. Now suppose it is true for $\mathcal{C}_f$. Then, $(\mathcal{C}_f)_{if+1} = 0$ for $i <f+1$. This implies the same is true for $W_{f+1}$. By symmetry, the same holds for the $f+1$-th row of $\mathcal{C}_f$. The same holds true for $W_{f+1}$ by definition. Hence, both $\mathcal{C}_f$ and $W_{f+1}$, and therefore $W_{f+1}^T$, are block diagonal with block dimensions $1, \cdots, 1, n-f$. This then implies the same is true for $\mathcal{C}_{f+1}$. The claim follows.

As established above, every matrix $W_f$ is full rank and therefore invertible. This implies $W$ is invertible. Hence, the matrix
\begin{align*}
  \mathcal{C}'^{-1} = W \mathcal{C}^{-1} W^T
\end{align*}
is well-defined. Since $\mathcal{C}'$ is block diagonal with block dimensions $1, \cdots, 1, n-F$, the same is true for $\mathcal{C}'^{-1}$. As per~\cref{eq:coupling_transform}, this transformation on $\mathcal{C}^{-1}$ is implemented precisely by performing the linear canonical transformation~\cref{eq:lin_can} with this matrix $W$. 

We now show the other consequences of implementing the linear canonical transformation with $W$, including how it keeps free modes explicit. We observe that for $i > F$,
\begin{align*}
  \Phi_i = \sum_{j}^{} W^{-1}_{ij} \Phi_j' = \Phi_i',
\end{align*}
where $\vec \Phi'\equiv W \vec \Phi$ is the transformed flux operators. Hence, the transformed fluxes include the original nonfree modes' fluxes. This means by removing the free modes in the Hamiltonian in terms of the transformed modes, the Hamiltonian on the original nonfree modes can be explicitly obtained. Furthermore, the junction term in the Hamiltonian is preserved since junction modes can never be free.

As per~\cref{eq:coupling_transform}, the transformation on the fluxes implies the flux coupling matrix is transformed as follows:
\begin{align*}
  M_0 \mapsto (W^T)^{-1} M_0 W^{-1}.
\end{align*}
That is, the same Gaussian elimination performed on $\mathcal{C}$ in~\cref{eq:gauss_elim} is performed on $M_0$. Since the free modes have no potential, this procedure preserves $M_0$. We can see this explicitly:
\begin{align*}
  (W_f M_0 W_f^T)_{ij} & = \sum_{k,l=1}^n (W_f)_{ik} (M_0)_{kl} (W_f)_{jl}\nonumber\\
  & = \sum_{k=i,f, l= j,f} (W_f)_{ik} (M_0)_{kl} (W_f)_{jl}\nonumber\\
  & = (W_f)_{ii} (M_0)_{ij} (W_f)_{jj}\nonumber\\
  & = (M_0)_{ij}
\end{align*}
since the $f$-th row and column of $M_0$ is zero. Hence,
\begin{align*}
  (W^T)^{-1} M_0 W^{-1} = M_0,
\end{align*}
as claimed above. Similarly, the external flux coupling matrix is transformed as
\begin{align*}
  N \mapsto (W^T)^{-1} N,
\end{align*}
but
\begin{align*}
  (W_f N)_{ij} & = \sum_{k=1}^{n} (W_f)_{ik} N_{kj}\nonumber\\
  & = (W_f)_{ii} N_{ij}\nonumber\\
  & = N_{ij},
\end{align*}
so $(W^T)^{-1} N =N$. This means the first $F$ transformed modes are also free modes, and the rest of the modes have the same flux and external flux couplings.

Furthermore, we claim that the charge coupling between nonfree modes is preserved. We first observe $W_f$ is an involution, that is, $W_f^{-1} = W_f$. To see this, compute
\begin{align*}
  (W_f W_f)_{ii} & = \sum_{k=1}^{n} (W_f)_{ik} (W_f)_{ki} \nonumber\\
  & = (W_f)_{ii}^2\nonumber\\
  & = 1.
\end{align*}
On the other hand, for $i \neq j$,
\begin{align*}
  (W_f W_f)_{ij}& = \sum_{k=1}^{n} (W_f)_{ik} (W_f)_{kj} \nonumber\\
  & = (W_f)_{ii} (W_f)_{ij} + (W_f)_{if} (W_f)_{fj}\nonumber\\
  & = \delta_{jf}( (W_f)_{ii} (W_f)_{if} + (W_f)_{if} (W_f)_{ff})\nonumber\\
  & = \delta_{jf}( (W_f)_{if} - (W_f)_{if} )\nonumber\\
  & = 0,
\end{align*}
where the second to last line follows since $i\neq j$ and so $i \neq f$. Then, for $i,j > F$,
\begin{align*}
  (\mathcal{C}_f^{-1})_{ij} & = [(W_f^T)^{-1} \mathcal{C}_{f-1}^{-1} W_f^{-1}]_{ij}\nonumber\\
  & = (W_f^T \mathcal{C}_{f-1}^{-1} W_f)_{ij}\nonumber\\
  & = \sum_{k,l=1}^{n} (W_f)_{ki} (\mathcal{C}^{-1}_{f-1})_{kl} (W_f)_{lj}\nonumber\\
  & = (W_f)_{ii} (\mathcal{C}^{-1}_{f-1})_{ij} (W_f)_{jj}\nonumber\\
  & = (\mathcal{C}^{-1}_{f-1})_{ij}.
\end{align*}
Hence, the submatrices of $\mathcal{C}^{-1}, \mathcal{C}'^{-1}$ corresponding to indices greater than $F$ are the same. This implies the claim.

In summary, the Hamiltonian in terms of the transformed modes takes the form
\begin{align*}
H & = \frac{1}{2} (\vec Q' - C_V' \vec V)^T  \mathcal{C}'^{-1} (\vec Q' - C_V' \vec V) \nonumber\\
  & - \sum_{i=1}^{n_J} E_{J,i} \cos \left(\frac{2 \pi}{\Phi_0} \Phi_i' \right) \nonumber\\
  & + \frac{1}{2} \vec \Phi'^T M_0 \vec \Phi' + \vec \Phi'^T N \vec \Phi_x,
\end{align*}
where
\begin{align*}
  \vec Q' \equiv (W^T)^{-1} \vec Q, \, C_V' \equiv (W^T)^{-1} C_V, \, \vec \Phi' \equiv W \vec \Phi.
\end{align*}
$H$ describes a system of $n$ modes with $F$ free modes which are each independent of all other modes, including other free modes. The Hamiltonian of the original nonfree modes can therefore be extracted by eliminating all the terms in $H$ corresponding to free mode charges or fluxes. Using the facts above, the extracted Hamiltonian is given by
\begin{align*}
  H_{\backslash F} = & \frac{1}{2} \vec Q_{\backslash F}'^T \mathcal{C}_{\backslash F}^{-1} \vec Q_{\backslash F}' - \sum_{i=1}^{n_J} E_{J,i} \cos \left(\frac{2 \pi}{\Phi_0} \Phi_i \right) \nonumber\\
  & + \frac{1}{2} \vec \Phi_{\backslash F}^T (M_0)_{\backslash F} \vec \Phi_{\backslash F} + \vec \Phi_{\backslash F}^T N_{\backslash F} \vec \Phi_x \nonumber \\
  & - [(C_V')_{\backslash F}\vec V]^T \mathcal{C}^{-1}_{\backslash F} \vec Q_{\backslash F}' ,
\end{align*}
where $_{\backslash F}$ means that we remove the components corresponding to free modes. Note that for $N, C_V'$ this means removing the rows corresponding to free modes. Also, we omit the term proportional to the identity coming from the quadratic term in the voltages $\vec V$ since this does not affect the physics. Lastly, the drive term $\propto V$ is as written because of the following argument. The term before removing free modes is given by
\begin{align*}
  -(C_V'\vec V)^T \mathcal{C}'^{-1} \vec Q'.
\end{align*}
Then, removing the free modes is equivalent to removing the first $F$ columns of $\mathcal{C}'^{-1}$ and the first $F$ entries of $\vec Q'$. However, since $\mathcal{C}'^{-1}$ is block diagonal, after removing the first $F$ columns, the first $F$ rows are all zero. Hence, the first $F$ rows can also be removed, as well as the first $F$ rows of $C_V'$. However, as shown above, the remaining submatrix of $\mathcal{C}'^{-1}$ is the same as that of $\mathcal{C}^{-1}$. Thus the expression shown above follows. 

Upon closer inspection, the only difference~\footnote{There is also the subtlety that $\vec Q_{\setminus F}$ is now replaced by $\vec Q'_{\setminus F}$, which are both conjugate to $\vec \Phi_{\setminus F}$ and hence the Hamiltonians are unitarily equivalent by~\cref{lem:lin_can}. } between the final expression for $H_{\backslash F}$ and simply removing the terms in $H$ corresponding to free modes is indeed only the $C_V'$ matrix in the voltage source term, which involves the $W$ matrix. This establishes our claim above. This result is also intuitive from the standpoint of Lagrangian mechanics. Since the flux operator for a free mode does not appear in the Hamiltonian, it is a cyclic coordinate. Thus, the corresponding momentum, which is a linear combination of charges, is conserved, and we can simply replace it by a constant.

Note that we can also bypass the iterative Gaussian elimination approach and instead directly tailor $W$ to block diagonalize $\mathcal{C}$. Suppose $\mathcal{C}$ has the following block diagonal structure between the free and non-free subspaces:
\begin{align*}
  \mathcal{C} = 
  \begin{pmatrix}
    A & X \\
    X^T & B
  \end{pmatrix}.
\end{align*}
Then, we can define
\begin{align*}
  (W^T)^{-1} = 
  \begin{pmatrix}
    I_F & 0_{F \times (n-F)}\\
    - X^T A^{-1} & I_{n-F} 
  \end{pmatrix}
\end{align*}
so that
\begin{align*}
  (W^T)^{-1} \mathcal{C} W^{-1} = 
  \begin{pmatrix}
    A & 0_{F \times (n-F)}\\
    0_{(n-F) \times F} & -X^T A^{-1} X + B
  \end{pmatrix}.
\end{align*}
That this choice of $W$ has the other desired properties follows from similar arguments given above and the formula for the inverse of a block matrix.

\section{Proof of Williamson's Theorem and Its Extension \label{app:williamson_proof}}
We here prove both Williamson's theorem and its extension, that is,~\cref{thm:williamson} and~\cref{cor:block_williamson}. We prove Williamson's theorem since the proof of its extension is based on the proof of the theorem we give here. For convenience we restate each:
\begin{thm*}
  (Williamson~\cite{williamson1936algebraic}). Let $M \in \R^{2n \times 2n}$ be a positive definite matrix. Then there exists a symplectic matrix $S \in \mathrm{Sp}(2n)$ and a diagonal matrix $\Lambda \in \R^{n\times n}$ with positive diagonal entries such that
  \begin{align*}
    S^T M S = \mathrm{diag}(\Lambda, \Lambda).
  \end{align*}
\end{thm*}
\begin{proof}
  (Adapted from~\cite{idel2017perturbation}) The proof involves an elegant interplay between real and complex numbers. Since $M$ is positive definite, we can define its positive square root $M^{1/2}$ and its inverse $M^{-1/2}$. Consider the complex matrix $i M^{-1/2} \Omega M^{-1/2}$. Since $M^{-1/2}$ is real symmetric, $i M^{-1/2} \Omega M^{-1/2}$ is complex Hermitian, where again $\Omega$ is the symplectic form:
  \begin{align*}
    \Omega \equiv
    \begin{pmatrix}
      0_n & I_n \\
      -I_n & 0_n
    \end{pmatrix}.
  \end{align*}
  Thus, we can find a basis of eigenvectors of $i M^{-1/2} \Omega M^{-1/2}$. However, we notice that if $v \in \mathbb{C}^{2n}$ is an eigenvector with eigenvalue $\lambda \in \mathbb{R}$, then
  \begin{align*}
    i M^{-1/2} \Omega M^{-1/2} v^* & = (- i M^{-1/2} \Omega M^{-1/2} v)^* \nonumber\\
    & = (-\lambda v)^* \nonumber\\
    & = - \lambda v^*,
  \end{align*}
where the asterisk refers to elementwise complex conjugation. Thus, $-\lambda$ is also an eigenvalue of $i M^{-1/2} \Omega M^{-1/2}$ with eigenvector $v^*$, which has to be orthogonal to $v$ if $\lambda \neq 0$. However, $\lambda$ cannot be zero since that would imply $\Omega$ has a nontrivial null space, which contradicts the fact $\det(\Omega) = 1$. Hence, we have a basis of normalized eigenvectors of $i M^{-1/2} \Omega M^{-1/2}$: $B =\{ v_1, \cdots, v_n,  v_1^*, \cdots, v_n^*\}$ . Now consider the real matrix 
$$O \equiv (x_1, \cdots, x_n, y_1, \cdots, y_n) \in \mathbb{R}^{2n \times 2n},$$
where
$$x_j \equiv \frac{v_j + v_j^*}{\sqrt{2}},\, y_j \equiv i \frac{v_j - v_j^*}{\sqrt{2}}.$$
By the orthonormality of $B$ we conclude that $O$ is a real orthogonal matrix. Now, define the real matrix $S \equiv M^{-1/2} O \tilde D^{-1/2}$ where $\tilde D \equiv \mathrm{diag}(D, D)$ and $D \equiv \mathrm{diag}(\lambda_1, \cdots, \lambda_n)$ is the diagonal of the positive eigenvalues of $i M^{-1/2} \Omega M^{-1/2}$, in the order corresponding to $B$. Without loss of generality, all the positive eigenvalues correspond to $v_j$ and the negative to $v_j^*$. Then, we compute:
\begin{align*}
  S^T \Omega S & = \tilde D^{-1/2} O^T M^{-1/2} \Omega M^{-1/2} O \tilde D^{-1/2} \nonumber\\
  & = -i\tilde D^{-1/2} O^T U \mathrm{diag}(D, -D) U^\dagger O \tilde D^{-1/2},
\end{align*}
where $U = (v_1, \cdots, v_n, v_1^*, \cdots, v_n^*)$ is the complex unitary that diagonalizes $i M^{-1/2} \Omega M^{-1/2}$. Next, we compute
\begin{align*}
  O^T U & = 
  \begin{pmatrix}
    x_1^T\\
    \vdots \\
    x_n^T \\
    y_1^T\\
    \vdots \\
    y_n^T
  \end{pmatrix}
  (v_1, \cdots, v_n, v_1^*, \cdots, v_n^*) \\
  & = \frac{1}{\sqrt{2}} 
  \begin{pmatrix}
    1_n & 1_n\\
    -i1_n & i1_n
  \end{pmatrix},
\end{align*}
and so
\begin{align*}
& S^T \Omega S \nonumber\\
& = \frac{-i}{2} \tilde D^{-1/2} 
\begin{pmatrix}
1_n & 1_n\\
-i1_n & i1_n
\end{pmatrix} 
\begin{pmatrix}
  D & 0_n \\
  0_n & -D
\end{pmatrix}
\begin{pmatrix}
1_n & i1_n\\
1_n & -i1_n
\end{pmatrix} \tilde D^{-1/2}\nonumber\\
& = \frac{-i}{2}  \tilde D^{-1/2} 
\begin{pmatrix}
1_n & 1_n\\
-i1_n & i1_n
\end{pmatrix}
\begin{pmatrix}
D & iD\\
-D & iD
\end{pmatrix} \tilde D^{-1/2}\nonumber \\
& = \frac{-i}{2} \begin{pmatrix}
D^{-1/2} & 0_n\\
0_n & D^{-1/2}
\end{pmatrix}
\begin{pmatrix}
0_n & 2iD \\
-2iD & 0_n
\end{pmatrix}
\begin{pmatrix}
D^{-1/2} & 0_n\\
0_n & D^{-1/2}
\end{pmatrix}\nonumber\\
& = \frac{-i}{2} 
\begin{pmatrix}
0_n & 2i D^{1/2}\\
-2i D^{1/2}&0_n
\end{pmatrix}
\begin{pmatrix}
D^{-1/2} & 0_n\\
0_n & D^{-1/2}
\end{pmatrix} \nonumber\\
& = \Omega.
\end{align*}
Thus, $S$ is symplectic. It is clear that $S^T M S = \tilde D^{-1}$, as desired.
\end{proof}

We now state and prove the extension.
\begin{cor*}
  If $M \in \R^{2n \times 2n}$ is positive definite and is block diagonal with two $n\times n$ blocks, there exists a block diagonal symplectic matrix 
  \begin{align*}
    S  = 
    \begin{pmatrix}
      S_n & 0_n\\
      0_n & (S_n^T)^{-1}
    \end{pmatrix}\in \mathrm{Sp}(2n)
  \end{align*}
and a diagonal matrix $\Lambda \in \R^{n\times n}$ with positive diagonal entries such that
  \begin{align*}
    S^T M S = \mathrm{diag}(\Lambda, \Lambda).
  \end{align*}
\end{cor*}
\begin{proof}
  We will roughly follow the previous proof. Denote our initial matrix by 
$$M = 
\begin{pmatrix}
A & 0_n\\
0_n & B
\end{pmatrix},$$
where $A, B \in \mathbb{R}^{n \times n}$. Now, since $M$ is positive definite, so are $A,B$. We can then write 
$$i M^{-1/2} \Omega M^{-1/2} = i
\begin{pmatrix}
0_n & A^{-1/2} B^{-1/2}\\
- B^{-1/2}A^{-1/2} & 0_n
\end{pmatrix}.$$
Consider a normalized real eigenvector of the real positive definite matrix $B^{-1/2} A^{-1} B^{-1/2}$, $w \in \mathbb{R}^{n}$ with eigenvalue $\lambda > 0$. Define the complex vector $w' \equiv \frac{i}{\sqrt{\lambda}} A^{-1/2} B^{-1/2} w \in \C^n$, which is also normalized with respect to the norm for $\C^n$. Then, we have 
\begin{align*}
& i M^{-1/2} \Omega M^{-1/2}
\begin{pmatrix}
w' \\
w
\end{pmatrix}\nonumber\\
&  = i
\begin{pmatrix}
0_n & A^{-1/2} B^{-1/2}\\
- B^{-1/2}A^{-1/2} & 0_n
\end{pmatrix}
\begin{pmatrix}
w' \\
w
\end{pmatrix}\nonumber\\
& = i
\begin{pmatrix}
A^{-1/2} B^{-1/2} w \\
-B^{-1/2} A^{-1/2} w'
\end{pmatrix}\nonumber\\
& =
i
\begin{pmatrix}
A^{-1/2} B^{-1/2} w \\
-\frac{i}{\sqrt{\lambda}} B^{-1/2} A^{-1} B^{-1/2} w
\end{pmatrix}\nonumber\\
& =
\begin{pmatrix}
i A^{-1/2} B^{-1/2} w \\
\frac{1}{\sqrt{\lambda}} \lambda w
\end{pmatrix}
= \sqrt{\lambda}
\begin{pmatrix}
w' \\
 w
\end{pmatrix}.
\end{align*}
Thus, $\tilde w \equiv 
\begin{pmatrix}
w'\\
w
\end{pmatrix}$ is a eigenvector of $i M^{-1/2} \Omega M^{-1/2}$ with eigenvalue $\sqrt{\lambda}$. By the proof of Williamson's theorem, 
$$\tilde w^* \equiv 
\begin{pmatrix}
-w'\\
w
\end{pmatrix}$$
is an eigenvector with eigenvalue $-\sqrt{\lambda}$. Now, given any two eigenvectors $w_i, w_j$ of $B^{-1/2} A^{-1} B^{-1/2}$, 
$$\tilde w_i^\dagger \tilde w_j = \frac{1}{\sqrt{\lambda_i \lambda_j}} w_i^T B^{-1/2}A^{-1} B^{-1/2} w_j + w_i^T w_j = 0$$
and similarly 
$$\tilde w_i^\dagger \tilde w_j^* = 0.$$
Thus, from a basis of real eigenvectors $\{ w_1 , \cdots , w_n \}$ of the real symmetric matrix $B^{-1/2} A^{-1} B^{-1/2}$, we obtain a basis of complex eigenvectors $\{\tilde w_1, \cdots \tilde w_n, \tilde w_1^* , \cdots, \tilde w_n^* \}$ of complex Hermitian matrix $i M^{-1/2} \Omega M^{-1/2}$, up to normalization. Following the above proof, we conclude that the real vectors $x_j, y_j$ are of the form
$$x_j =
\begin{pmatrix}
\vec 0\\
w
\end{pmatrix},$$  $$y_j =
\begin{pmatrix}
i w' \\
\vec 0
\end{pmatrix},$$
where $\vec 0 \in \mathbb{R}^n$ is the zero vector. Hence, the matrix $O$ is of the form
$$O =
\begin{pmatrix}
0_n & O_1\\
O_2 & 0_n
\end{pmatrix},$$
\newpage
\noindent and so is the $S$ obtained from Williamson's theorem:
\begin{align*}
  S & \equiv M^{-1/2} O \tilde D^{-1/2} \nonumber\\
  & = 
\begin{pmatrix}
0_n & A^{-1/2} O_1 D^{-1/2}\\
B^{-1/2} O_2 D^{-1/2} & 0_n
\end{pmatrix}.
\end{align*}
However, defining 
\begin{align*}
S' \equiv S \Omega = 
\begin{pmatrix}
- A^{-1/2} O_1 D^{-1/2} & 0_n\\
0_n & B^{-1/2} O_2 D^{-1/2}
\end{pmatrix}, 
\end{align*}
we find
$$S'^T M S' = \Omega^T \tilde D^{-1} \Omega = \tilde D^{-1}.$$
Since symplectic transformations form a group, $S'$ is the desired block diagonal symplectic matrix. The relationship between the block matrices of $S'$ follows directly from the symplectic condition $S'^T \Omega S' = \Omega$.
\end{proof}

\bibliography{ref}

\end{document}